\documentclass[letterpaper]{vldb}

\usepackage{subfigure}
\usepackage{graphicx}
\usepackage{balance}
\usepackage{xspace}
\usepackage{multirow}

\usepackage{times}
\usepackage{color}

\usepackage{url}
\usepackage{hyperref}
\usepackage{balance}
\usepackage{enumerate}
\usepackage[linesnumbered,vlined,ruled]{algorithm2e}

\newtheorem{theorem}{Theorem}

\newtheorem{lemma}[theorem]{Lemma}

\newtheorem{observation}[theorem]{Observation}
\newtheorem{corollary}[theorem]{Corollary}
\newtheorem{example}[theorem]{Example}

\newtheorem{assumption}[theorem]{Assumption}

\DeclareMathOperator{\cost}{cost}
\DeclareMathOperator{\optcost}{opt-cost}
\DeclareMathOperator{\actcost}{act-cost}
\DeclareMathOperator{\extcost}{ext-cost}
\DeclareMathOperator{\relcost}{rel-cost}
\DeclareMathOperator{\nil}{nil}
\DeclareMathOperator{\pivot}{pivot}

\DeclareMathOperator{\new}{new}
\DeclareMathOperator{\old}{old}

\DeclareMathOperator{\E}{E}

\DeclareMathOperator{\Cov}{Cov}

\begin{document}

%\title{Towards Multi-tenant machine learning Infrastructures}
%\title{Towards ``Machine Learning as a Service'': \\
%The Multi-tenant Case}

\title{A Note On Operator-Level Query Execution Cost Modeling}

\author{
\alignauthor
Wentao Wu\\
\vspace{0.5em}
\affaddr{Microsoft Research, Redmond}\\
\email{wentao.wu@microsoft.com}
}

\maketitle

\begin{abstract}
External query execution cost modeling using query execution feedback has found its way in various database applications such as admission control and query scheduling.
Existing techniques in general fall into two categories, \emph{plan-level} cost modeling and \emph{operator-level} cost modeling.
It has been shown in the literature that operator-level cost modeling can often significantly outperform plan-level cost modeling.
In this paper, we study operator-level cost modeling from a robustness perspective.
We address two main challenges in practice regarding limited execution feedback (for certain operators) and mixed cost estimates due to the use of multiple cost modeling techniques.
We propose a framework that deals with these issues and present a comprehensive analysis of this framework.
We further provide a case study to demonstrate the efficacy of our framework in the context of index tuning, which is itself a new application of external cost modeling techniques.
\end{abstract}

\section{Introduction}

In recent years, there is substantial work that tries to provide more accurate estimate for query execution cost~\cite{AhmadDAB11-edbt,AkdereCRUZ12-brown-icde,DugganCPU11,Ganapathi-berkeley09,Li12Robust,WuCHN13,WuCZTHN13}.
Unlike early work that mainly focuses on improving cost estimates inside the optimizer (prominently, via improved cardinality estimates), this line of work chose to sit outside the query optimizer.
That is, the cost models proposed in this work are external to the query optimizer and have no impact on its plan choice.
The effectiveness of these external cost models has been demonstrated in various applications such as admission control~\cite{Tozer-QCop} and query scheduling~\cite{Ahmad-QShufflerVLDBJ11}.
However, their applications in other areas, such as query optimization and index tuning, remain limited.

Existing external cost modeling approaches can in general fall into two categories: \emph{plan-level} and \emph{operator-level} modeling.
In plan-level modeling, one first converts the entire query plan into a feature vector and then trains a machine learning model that learns the relationship between the feature vector and the corresponding plan execution time.
In operator-level modeling, one applies the same idea to -- instead of the entire plan -- each individual operators.
The estimated operator costs are then combined (e.g., summed up) to generate the estimated cost of the entire plan.
As was shown in previous work~\cite{AkdereCRUZ12-brown-icde,Li12Robust}, operator-level modeling can often significantly outperform plan-level modeling.
There are, however, three major challenges in operator-level modeling.

\vspace{0.5em}
\noindent\textbf{(Challenge 1: Appropriate Training)}
Operator-level models require appropriate training before they could be effective.
This then raises a question about training data --- on which dataset should we train the models?
A premise underlying these (and any) learning-based models is that the training set should be \emph{representative} for the workload.
As was demonstrated in previous work~\cite{AkdereCRUZ12-brown-icde}, learning-based models are much better if they are trained and tested over the same database and workload (see Section~\ref{sec:framework:challenge} for more details).
Therefore, the training data set needs to be harvested from past execution history over the same database.

\vspace{0.5em}
\noindent\textbf{(Challenge 2: Limited Execution Feedback)}
Previous work on learning-based models presumes the availability of sufficient training data, which is often not the case in a real production setting.
That is, it rarely happens that abundant execution information about a workload is available, especially for interactive workloads that often contain lots of ad-hoc queries.

\vspace{0.5em}
\noindent\textbf{(Challenge 3: Mixed Cost Estimates)}
Operator-level modeling requires training a model for each operator that may appear in a query execution plan.
Given the scarcity of past execution feedback, it is unlikely that each model can receive similar amount of training data.
For instance, for a database that contains very few or no indexes, the query execution plans may contain much fewer index-based nested-loop joins compared to hash joins.
Therefore, in general we may have little (or even no) execution feedback for certain operators.
Training models for such operators is impossible and a natural solution is fallback to optimizer's cost estimates.
Then, however, \emph{how do we combine cost estimates made by different modeling approaches}?
Note that the cost estimates may even have different semantics.
For example, external cost models usually target elapsed time or resource consumption such as CPU time, whereas optimizer's cost estimates are often more ``abstract'' and may not have concrete semantics.
In such situations, we cannot directly combine them (e.g., by adding them up).
%The problem of combining mixed cost estimates exists as long as we use different cost modeling approaches.

\vspace{0.5em}
With the above challenges in mind, in this paper we propose a simple yet general framework that operates on top of limited execution feedback, which consists of three steps.
First, we identify a set of \emph{backbone} operators that serve as workhorse for almost all workload queries yet have relatively abundant execution feedback.
In this work, we focus on using leaf operators, such as table scans, index scans, and index seeks, as backbone operators, though it is straightforward to include other operators such as joins.
The reason for this choice is obvious: Regardless of which query plan the optimizer chooses for a given query, it always accesses the same tables, whereas the internal operators such as joins can be quite different (e.g., different join orders may be chosen).
Feedback information such as input/output cardinality is therefore more likely reusable across query plans for leaf operators, which results in more reliable models.
Second, we build an external cost model for each leaf operator using existing techniques, and use the models to estimate costs for leaf operators in each query plan.
Third, we combine the cost estimates made by external cost models for the leaf operators with the optimizer's cost estimates for the internal operators, using a straightforward yet principled technique.

%Our experimental evaluation shows that this framework works well over most of the databases and workloads we have tested, though on certain databases it has no obvious advantage over the default approach that simply uses optimizer's cost estimates.
To understand when this framework can work and when it may not work, we further conduct a theoretical analysis.
Our analysis reveals that the effectiveness of this framework depends on how overwhelming the (diversity or variation in the) amount of work done by the backbone operators compared with the other operators.

%To summarize, this paper makes the following contributions:
%\begin{enumerate}[(\textbf{C}1)]
%    \item We take a first step towards the promising direction of utilizing query execution feedback to improve index tuning.
%    \item We propose a general framework that can handle sparse and inconsistent feedback.
%    \item We implement a specific version of the general framework and evaluate it via extensive experiments.
%    \item We study the framework from a theoretical perspective and extend our analysis into continuous index tuning.
%\end{enumerate}

%\noindent\textbf{(Limitations)}
%Not all these challenges can be addressed in certain applications.
%For example, if our goal is to predict query execution cost such as elapsed time or CPU time, then the issue of mixed cost estimates may not be resolved without a clear understanding of the semantics of the optimizer's cost model.
%However, if our goal is only to \emph{compare} the costs of different query plans, as in many applications such as query optimization or index tuning, then we do not need to worry about the semantics of the cost estimates.
%This gives us the freedom of converting one type of cost estimates to another.

\vspace{0.5em}
\noindent\textbf{(Paper Organization)}
The rest of the paper is organized as follows.
We study the practical challenges of utilizing execution feedback in detail, and propose a general framework that addresses these challenges (Section~\ref{sec:framework}).
We next present an analysis of the framework in Section~\ref{sec:analysis}.
In Section~\ref{sec:experiments}, we further present a case study of the applicability of the proposed framework in the context of index tuning.
We summarize related work in Section~\ref{sec:relatedwork} and conclude the paper in Section~\ref{sec:conclusion}.

%\vspace{-0.5em}
\section{The Framework}\label{sec:framework}

The high-level idea of utilizing execution feedback to build external, operator-level cost models is clear.
However, as was discussed in the introduction, there are several practical challenges that may limit the effectiveness of the cost models.
In the following, we start by understanding these challenges, in retrospect, from lessons learned by previous work.
We then propose a general framework that addresses these challenges and present a specific implementation that will be evaluated in our experiments.

\subsection{Challenges of Utilizing Feedback}\label{sec:framework:challenge}

Existing external cost modeling techniques all rely on certain degree of ``learning.''
In the literature, people have been trying either plan-level or operator-level cost modeling techniques.
It has been shown that operator-level modeling is superior to plan-level modeling when workload drifts~\cite{AkdereCRUZ12-brown-icde}.
Nonetheless, operator-level modeling is still sensitive to the training data in use.

\vspace{-0.5em}
\paragraph*{Lessons from Previous Work}
To illustrate this sensitivity, we studied the results reported by~\cite{Li12Robust}, which, as far as we know, represents the state-of-the-art operator-level modeling.
The authors of~\cite{Li12Robust} compared various operator-level modeling techniques in their experimental evaluation.
We observe the following two facts from their comparison results.

First, if we train and test the models using the same workload, the estimation accuracy is quite good even in the presence of cardinality estimation errors.
The authors of~\cite{Li12Robust} used the TPC-H benchmark to train and test the models, and found that more than 80\% of the test cases have a ratio error below 1.5 (i.e., 50\% relative error), and the percentage can be improved to 90\% if true cardinality information is available.

Second, if we train and test the models over different workloads, the estimation accuracy drops dramatically.
The authors of~\cite{Li12Robust} used the TPC-H benchmark to train the models and tested the models on the TPC-DS benchmark as well as two real workloads. %(Real-1 and Real-2).
Based on their reported results, only 30\% of the test cases now have a ratio error below 1.5 on TPC-DS, and the percentage is about 40\% on the two real workloads.
However, if true cardinality information is available, the percentage can be improved to 70\% over all the three workloads.
Unfortunately, true cardinality information is usually not available for queries in the testing set, and previous work has reported worse results if true/estimated cardinalities are used in training whereas estimated/true cardinaities are used in testing~\cite{AkdereCRUZ12-brown-icde}.

The above observations suggest that operator-level modeling approaches are sensitive to the training set.
Only if the testing queries are drawn from the same training workload should we expect good estimation accuracy.
Cardinality estimation error remains one of the major factors that prevent better generalization of the models accross different workloads: Fixing cardinality errors can give us a lift between 30\% and 40\% in terms of the ratio-error metric.
Interestingly, within the same workload cardinality errors do not have significant impact on model estimation accuracy.
This has also been evidenced by other previous work~\cite{AkdereCRUZ12-brown-icde}.
One intuition for this is that cardinality estimation errors depend on the difference between query optimizer's presumed data distribution and the actual data distribution.
As long as this difference is consistent across different queries over the same database, cardinality estimation errors are like systematic biases that could be accounted for by the models learnt from training data.
This is, however, usually not the case across workloads over different databases, and therefore the biases have to be learnt again by the model.
%Nevertheless, even if cardinality estimation errors were fixable, there is still a 20\% gap (70\% vs. 90\%) that suggests modeling issues beyond cardinality estimation.

\vspace{-0.5em}
\paragraph*{Discussion}
Lessons from previous work suggest that it may be too ambitious to expect an external cost model that works ``everywhere.''
A more realistic approach is to learn a cost model for a fixed database and workload, which, as has been demonstrated by various previous work, can outperform optimizer's cost estimates (with some naive scaling).
An even more interesting observation here by previous work is that the specific machine learning models do not matter too much on a fixed database and workload, as long as they capture both the linear and nonlinear factors in cost modeling~\cite{AkdereCRUZ12-brown-icde,Li12Robust,WuCZTHN13}.

Moreover, this approach is sensitive to the amount of execution feedback we possess, the collection of which may be expensive in many situations.
In the worst case, we may not even have any feedback so it is equivalent to using optimizer's cost estimates.
%Our work in this paper is orthogonal to this, though.
%In our following study, we assume a fixed set of execution feedback.

\vspace{-0.5em}
\paragraph*{Insufficient Feedback and Mixed Cost Estimates}

The above discussion naturally raises questions about insufficient feedback.
It is quite likely that we do not have enough feedback to train models for certain operators.
A natural solution is to use optimizer's cost estimates for such operators.
This results in a new challenge of combining two different types of cost estimates.

\begin{figure}
\centering
    \includegraphics[width=0.8\columnwidth]{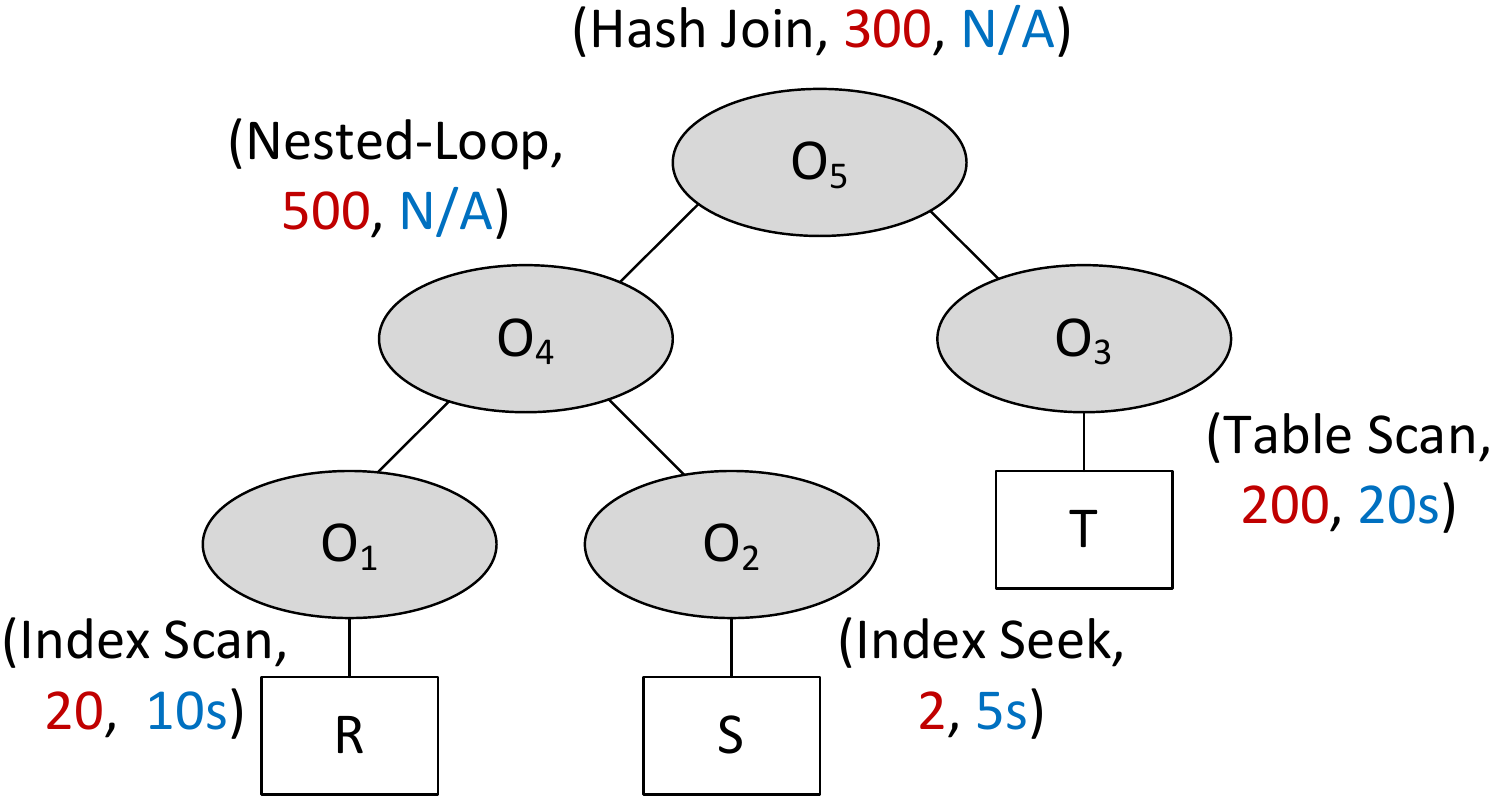}
\vspace{-1em}
\caption{A running example.}
\label{fig:running-example}
\vspace{-1em}
\end{figure}

%\vspace{-0.5em}
\begin{example}
To understand the issue of mixed cost estimates better, we present an annotated query execution plan as a running example in Figure~\ref{fig:running-example}.
Here $R$, $S$, and $T$ are tables, whereas $O_1$ to $O_5$ are physical operators.
We annotate each operator with its type and estimated cost.
In this example, the cost estimates for $O_1$, $O_2$, and $O_3$ come from external cost models, which indicate the estimated CPU time of these operators.
On the other hand, the cost estimates for $O_4$ and $O_5$ are made by the optimizer, which use some abstract metric and do not have concrete semantics.
\end{example}

%As a result, it remains for an application such as index tuning to figure out how to train an external cost model appropriately.
%In a real production environment, workload drift is quite likely.
%Therefore, it is important for a cost model to quickly adapt to a new workload pattern, given that the estimation accuracy heavily depends on the %assumption that testing queries are similar to training queries.

%On the other hand, index tuning introduces a new dimension that has not been studied by previous work on external cost modeling, namely, physical design change.
%All previous work assumed that the same fixed physical database is used in both training and testing.
%In the following, we present an empirical study that indicates the sensitivity of external cost modeling to physical design change.
%Specifically, we show that an external cost model trained over one index configuration may not be very effective over a different index configuration.

\subsection{A General Framework}

We propose a general framework to address the challenge of mixed cost estimates that naturally arises in the presence of insufficient execution feedback.
Algorithm~\ref{alg:combine-costs} presents the details.
We summarize the notation in Table~\ref{tab:notation}.

\begin{algorithm}
  \SetAlgoLined
  \KwIn{$P$, a query plan; $\mathcal{O}$, the set of operators with sufficient feedback; $\mathcal{M}$, the operator-level models built with feedback from $\mathcal{O}$.}
  \KwOut{$\cost(P)$, estimated cost of $P$.}
  \SetAlgoLined
  $\cost(P)\leftarrow 0$\;

  $o^{\pivot}\leftarrow PickPivot(\mathcal{O})$;  // Find the ``pivot'' (Algorithm~\ref{alg:pivot}).\\

  %// Estimate the cost for $P$.\\
  \ForEach{\emph{operator} $o\in P$}{
    \uIf{\emph{there is a model} $M\in\mathcal{M}$ \emph{for} $o$}{
        $\extcost(o)\leftarrow M(o)$\;
        $\cost(P)\leftarrow\cost(P)+\frac{\extcost(o)}{\actcost(o^{\pivot})}\cdot\optcost(o^{\pivot})$\;
    }\Else{
        $\cost(P)\leftarrow\cost(P)+\optcost(o)$\;
    }
  }

  \Return{$\cost(P)$}\;
  \caption{Combine mixed cost estimates.}
\label{alg:combine-costs}
\end{algorithm}

\begin{table}%[!htb]
\centering
\begin{tabular}{|l|l|}
\hline
Notation & Description \\
\hline
\hline
$o$ & An operator in the query plan\\
$\optcost(o)$ & The optimizer's estimated cost of $o$\\
$\actcost(o)$ & The actual execution cost of $o$\\
$\extcost(o)$ & The cost estimate of $o$ from external modeling\\
$o^{\pivot}$ & The pivot operator\\
\hline
\end{tabular}
\caption{Terminology and notation.}
\label{tab:notation}
\vskip -2ex
\end{table}

The main idea here for combining mixed cost estimates is simple.
We choose one ``pivot'' operator $o^{\pivot}$ from the operators with execution feedback (line 2).
We use the execution cost of $o^{\pivot}$ as a baseline, and compute the \emph{relative} cost
\begin{equation}
    \relcost(o)=\frac{\extcost(o)}{\actcost(o^{\pivot})}
\end{equation}
for any operator $o$ in the given plan $P$ where we have an external operator-level cost model.
We then scale the relative cost back using $\optcost(o^{\pivot})$.
For any operator in $P$ without sufficient feedback (i.e., we have not built a usable external cost model), we simply use the optimizer's cost estimate for it (lines 3 to 10).

\begin{example}
Continue with the running example in Figure~\ref{fig:running-example}.
Given the three operators $O_1$ to $O_3$ with available execution feedback, suppose that we choose $O_3$ as the pivot operator.
Assume that the external cost models for table scans, index scans, and index seeks are perfect, i.e., for any such operator $o$ we would have $\extcost(o)=\actcost(o)$.
The relative costs of $O_1$, $O_2$, and $O_3$ can then be easily computed:
$$\relcost(O_1)=\frac{\extcost(O_1)}{\actcost(O_3)}=\frac{\actcost(O_1)}{\actcost(O_3)}=\frac{10}{20}=0.5,$$
$$\relcost(O_2)=\frac{\extcost(O_2)}{\actcost(O_3)}=\frac{\actcost(O_2)}{\actcost(O_3)}=\frac{5}{20}=0.25,$$
and
$$\relcost(O_3)=\frac{\extcost(O_3)}{\actcost(O_3)}=\frac{\actcost(O_3)}{\actcost(O_3)}=\frac{20}{20}=1.$$
Meanwhile, the scaling factor is $\optcost(O_3)=200$.
Consequently, the adjusted estimated costs for $O_1$, $O_2$, and $O_3$ are $0.5\times 200 = 100$, $0.25\times 200 = 50$, and $1\times 200 = 200$.
Therefore, the final estimated cost for the example plan $P$ is
$\cost(P)=100+50+200+500+300=1150$.
\end{example}

The above procedure relies on the following assumption:
\begin{assumption}\label{assumption:consistency}
The cost estimates made by external cost models are comparable and the relative costs are well defined.
The same holds for optimizer's cost estimates.
\end{assumption}
In other words, Assumption~\ref{assumption:consistency} states that the semantics of cost estimates produced by the same model should be \emph{consistent}: If a cost model estimates CPU time for one operator then it should do the same for the others.
In theory, a cost model is not required to conform to this restriction.
For example, we could have a cost model that estimates the number of IO's for table scans whereas estimates the CPU time for hash joins.
We are not aware of such cost models in practice and do not consider them in this paper.
Moreover, if a cost model needs to produce a single number (instead of a vector of numbers) as its cost estimate at the operator level, then it has to be consistent; otherwise it is unclear how to combine the operator-level cost estimates at the plan level.

\vspace{-0.5em}
\paragraph*{Remarks}
Two remarks are in order. First, Algorithm~\ref{alg:combine-costs} does not specify how to pick the pivot operator $o^{\pivot}$. In theory, we could pick \emph{any} operator with execution feedback. However, it is clear that the choice of $o^{\pivot}$ has impact on the estimated cost, because we use $\optcost(o^{\pivot})$ as the scaling factor when combining with optimizer's cost estimates.
We study this impact in Section~\ref{sec:analysis}.

Second, Algorithm~\ref{alg:combine-costs} assumes that the external models are already built and fixed.
This should not be the case in reality.
As we accumulate execution feedback, it makes sense to train the models again.
Moreover, it is likely that operators lack of external cost modeling initially can catch up if enough feedback is available.
As a result, the inputs $\mathcal{O}$ and $\mathcal{M}$ will change dynamically.

\subsection{Implementation}\label{sec:framework:implementation}

Even if the inputs $\mathcal{O}$ and $\mathcal{M}$ are fixed, we still need to select the set $\mathcal{O}$.
In general, the selection depends on several factors:
\begin{itemize}
    \item First, the operators in $\mathcal{O}$ should have sufficient amount of execution feedback.
It is unlikely that we can train a promising model with little training data.
    \item Second, as cardinality information is used as one prominent feature in almost all existing operator-level cost modeling techniques, the operators in $\mathcal{O}$ should have relatively robust cardinality estimates.
    \item Third, the operators should cover significant amount of work performed by an arbitrary query plan. Otherwise, even if we can have perfect cost estimates for these operators, the impact on the overall cost estimate of the query plan is limited.
\end{itemize}

We therefore propose to focus on \emph{leaf} operators, including table scans, index scans, index seeks, and so on (depending on specific database systems), which meet the above three criteria well.
Nonetheless, there is no mandatory reason to exclude internal operators, though the decision is more workload-dependent.
For example, for workloads with little data correlation, cardinality estimation errors may not be severe and thus it is convincing to also include join operators into $\mathcal{O}$ if they have enough feedback.
In the rest of the paper, we call the operators in $\mathcal{O}$ the \emph{backbone} operators.

%as good cost estimates for these operators will have significant impact on plan cost estimation.
%Note that one could even use different backbone operators for different queries, if additional knowledge is available.
%Nevertheless, this is beyond the scope of the current paper.

\section{Analysis}\label{sec:analysis}

In this section, we present analysis of the framework illustrated in Algorithm~\ref{alg:combine-costs}.
We first discuss target performance metrics and then formalize the problem we will study.
Based on our problem formulation, we provide answers to the following questions:
\begin{itemize}
    \item When does the approach work and when may it not work?
    \item What can impact its performance and how can we improve?
\end{itemize}

\subsection{Performance Metrics}\label{sec:analysis:metrics}

The first question is how to measure the performance of our approach.
One could aim for reducing cost estimation errors, just like previous work on external cost modeling.
However, in many applications such as query optimization or index tuning, we are not very interested in the specific numbers returned by cost models,
because we only care about whether we can \emph{compare} query plans based on these numbers (i.e., estimated plan costs).
Therefore, as far as we can distinguish plans by their \emph{relative} costs we are satisfied.
For example, we perhaps only need to know that one plan is 20\% better/cheaper than the other one.
This suggests that we primarily consider the \emph{correlation} between the cost estimates and the actual costs of the plans.
Therefore, we use the well-known Pearson correlation coefficient (Pearson CC) as our performance metric. %which captures the correlation between estimated and actual plan costs.
%Since Pearson correlation coefficient is sensitive to outliers, we also use the more robust Spearman correlation coefficient, which is essentially the rank-based version of Pearson correlation coefficient.
%In our following analysis we will focus on Pearson correlation coefficient, as it is more analytically tractable.
%Since Pearson correlation coefficient is sensitive to outliers, in our experimental evaluation we report results for both Pearson and Spearman correlation coefficients.

\subsection{Problem Formulation}

Throughout this section, assume that we use leaf operators as the backbone operators, for which we have external cost models built using execution feedback.
We present the notation we will use in the following analysis in Table~\ref{tab:notation-analysis}.

\begin{table}[t]
\centering
\begin{tabular}{|l|l|}
\hline
Notation & Description \\
\hline
\hline
$\mathcal{L}$ & Leaf operators \\
$\mathcal{I}$ & Internal operators \\
\hline\hline
$P$ & Plan CPU time\\
$L$ & Leaf CPU time\\
$I$ & Internal CPU time\\
$\alpha$ & $\rho(L,I)$, Pearson CC between $L$ and $I$\\
$\sigma_L$ & Standard deviation of $L$\\
$\sigma_I$ & Standard deviation of $I$\\
$\eta$ & $\eta =\frac{\sigma_L}{\sigma_I}$\\
\hline\hline
$P'$ & Estimated plan cost\\
$L'$ & Estimated leaf cost\\
$I'$ & Estimated internal cost\\
$\sigma_{L'}$ & Standard deviation of $L'$\\
$\sigma_{I'}$ & Standard deviation of $I'$\\
$\eta'$ & $\eta =\frac{\sigma_{L'}}{\sigma_{I'}}$\\
$\beta$ & $\rho(L,I')$, Pearson CC between $L$ and $I'$\\
$\gamma$ & $\rho(I,I')$, Pearson CC between $I$ and $I'$\\
\hline\hline
$\rho$ & $\rho(P,P')$, Pearson CC between $P$ and $P'$\\
\hline
\end{tabular}
\caption{Notation used in the formal analysis.}
\label{tab:notation-analysis}
\vskip -2ex
\end{table}

We use $P$, $L$, and $I$ to represent the total CPU time spent on the whole plan, the leaf operators, and the internal operators, respectively.
Clearly, $P=L+I$, where
\begin{equation}
  L=\sum\nolimits_{o\in\mathcal{L}}\actcost(o)
\end{equation}
and
\begin{equation}
  I=\sum\nolimits_{o\in\mathcal{I}}\actcost(o).
\end{equation}
Similarly, $P'=L'+I'$, where, by Algorithm~\ref{alg:combine-costs},
\begin{equation}\label{eq:L-prim}
  L'=\sum\nolimits_{o\in\mathcal{L}}\frac{\extcost(o)}{\actcost(o^{\pivot})}\cdot\optcost(o^{\pivot})
\end{equation}
and
\begin{equation}
    I'=\sum\nolimits_{o\in\mathcal{I}}\optcost(o).
\end{equation}
To simplify our analysis, assume that the external cost models are perfect, namely, $\extcost(o)=\actcost(o)$ for any $o\in\mathcal{L}$.
(We will study the impact of cost modeling errors later.)
Moreover, define a constant
\begin{equation}\label{eq:lambda}
    \lambda = \frac{\optcost(o^{\pivot})}{\actcost(o^{\pivot})}.
\end{equation}
By Equation~\ref{eq:L-prim}, it follows that
\begin{equation}\label{eq:L-prim-new}
    L'=\lambda\cdot\sum\nolimits_{o\in\mathcal{L}}\actcost(o)=\lambda\cdot L.
\end{equation}

\subsection{Correlation Analysis}

We are interested in the Pearson correlation coefficient $\rho(P,P')$ between $P$ and $P'$.
Based on the previous formulation, we have
\begin{equation}\label{eq:rho-def}
    \rho=\rho(P,P')=\rho(L+I,L'+I')=\rho(L+I,\lambda\cdot L + I').
\end{equation}
With the notation in Table~\ref{tab:notation-analysis}, we can prove the following lemma.
%(The proofs of the theoretical results can be found in Appendix~\ref{sec:proofs}.)
\begin{lemma}\label{lemma:rho}
$\rho$ only depends on $\eta$, $\eta'$, $\alpha$, $\beta$, and $\gamma$. Specifically, %we have
\begin{equation}\label{eq:rho}
    \rho=\frac{\eta\eta'+\alpha\eta'+\beta\eta+\gamma}{\sqrt{\eta^2+2\alpha\eta+1}\cdot\sqrt{(\eta')^2+2\beta\eta'+1}}.
\end{equation}
\end{lemma}

\begin{proof}
By Equation~\ref{eq:rho-def}, we have
$$\rho=\rho(L+I,\lambda L + I')=\frac{\Cov(L+I,\lambda L + I')}{\sigma_{L+I}\cdot\sigma_{\lambda L + I'}}.$$
By the definition of covariance,
$$\Cov(L+I,\lambda L + I')=\E[(L+I)-\E(L+I)][(\lambda L + I')-\E(\lambda L + I')].$$
Using simple arithmetic calculation, we can obtain
$$\Cov(L+I,\lambda L + I')=\lambda\sigma_L^2+\lambda\cdot\Cov(L,I)+\Cov(L,I')+\Cov(I,I').$$
On the other hand, by the definition of variance,
$$\sigma_{L+I}^2=\E[(L+I)-\E(L+I)]^2=\sigma_L^2+2\cdot\Cov(L,I)+\sigma_I^2.$$
Similarly, we have
$$\sigma_{\lambda L+I'}^2=\lambda^2\sigma_L^2+2\lambda\cdot\Cov(L,I')+\sigma_{I'}^2.$$
Using the relationships
$$\Cov(L,I)=\rho(L,I)\sigma_L\sigma_I=\alpha\sigma_L\sigma_I,$$
$$\Cov(L,I')=\rho(L,I')\sigma_L\sigma_{I'}=\beta\sigma_L\sigma_{I'},$$
$$\Cov(I,I')=\rho(I,I')\sigma_I\sigma_{I'}=\gamma\sigma_I\sigma_{I'},$$
it then follows that
$$\Cov(L+I,\lambda L + I')=\lambda\sigma_L^2+\lambda\alpha\sigma_L\sigma_I+\beta\sigma_L\sigma_{I'}+\gamma\sigma_I\sigma_{I'},$$
$$\sigma_{L+I}^2=\sigma_L^2+2\alpha\sigma_L\sigma_I+\sigma_I^2,$$
$$\sigma_{\lambda L+I'}^2=\lambda^2\sigma_L^2+2\lambda\beta\sigma_L\sigma_{I'}+\sigma_{I'}^2.$$
As a result, we have
\begin{equation*}
    \rho =\frac{\lambda\sigma_L^2+\lambda\alpha\sigma_L\sigma_I+\beta\sigma_L\sigma_{I'}+\gamma\sigma_I\sigma_{I'}}{\sqrt{\sigma_L^2+2\alpha\sigma_L\sigma_I+\sigma_I^2}\cdot\sqrt{\lambda^2\sigma_L^2+2\lambda\beta\sigma_L\sigma_{I'}+\sigma_{I'}^2}}.
\end{equation*}
Dividing both the numerator and the denominator by $\sigma_I\sigma_{I'}$, %we have
\begin{equation*}
\rho=\frac{\lambda\frac{\sigma_L}{\sigma_I}\frac{\sigma_L}{\sigma_{I'}}+\lambda\alpha\frac{\sigma_L}{\sigma_{I'}}+\beta\frac{\sigma_L}{\sigma_I}+\gamma}{\sqrt{\big(\frac{\sigma_L}{\sigma_I}\big)^2+2\alpha\frac{\sigma_L}{\sigma_I}+1}\cdot\sqrt{\lambda^2\big(\frac{\sigma_L}{\sigma_{I'}}\big)^2+2\lambda\beta\frac{\sigma_L}{\sigma_{I'}}+1}}.
\end{equation*}
Since $\eta=\frac{\sigma_L}{\sigma_I}$ and $\eta'=\frac{\sigma_{L'}}{\sigma_{I'}}=\frac{\lambda\sigma_L}{\sigma_{I'}}$, it follows that
\begin{equation*}
\rho=\frac{\eta\eta'+\alpha\eta'+\beta\eta+\gamma}{\sqrt{\eta^2+2\alpha\eta+1}\cdot\sqrt{(\eta')^2+2\beta\eta'+1}}.
\end{equation*}
This completes the proof of the lemma.
\end{proof}

We can have several interesting observations based on Lemma~\ref{lemma:rho}.
First, we have the following lower bounds for $\rho$ that only depend on $\eta$ and $\eta'$ (Theorem~\ref{theorem:lower-bound} and Corollary~\ref{corollary:lower-bound-positive}).
\begin{theorem}\label{theorem:lower-bound}
Define a function
$$f(\eta,\eta')=\frac{\eta\eta'-\eta'-\eta -1}{(\eta+1)(\eta'+1)}=\frac{1-\frac{1}{\eta}-\frac{1}{\eta'}-\frac{1}{\eta\eta'}}{(1+1/\eta)(1+1/\eta')}.$$
For any $0\leq\eta,\eta'<\infty$, we have $\rho\geq f(\eta,\eta').$
\end{theorem}

\begin{proof}
We have $-1\leq \alpha,\beta,\gamma\leq 1$. Based on Equation~\ref{eq:rho},
$$\eta\eta'+\alpha\eta'+\beta\eta+\gamma\geq \eta\eta'-\eta'-\eta -1,$$
$$\sqrt{\eta^2+2\alpha\eta+1}\leq\sqrt{\eta^2+2\eta+1}=\eta+1,$$
$$\sqrt{(\eta')^2+2\beta\eta'+1}\leq\sqrt{(\eta')^2+2\eta'+1}=\eta'+1.$$
As a result, it follows that
$$\rho\geq\frac{\eta\eta'-\eta'-\eta -1}{(\eta+1)(\eta'+1)}=f(\eta,\eta').$$
This completes the proof the theorem.
\end{proof}

%\vspace{-4em}
\begin{corollary}\label{corollary:lower-bound-positive}
Define a function
$$g(\eta,\eta')=\frac{\eta}{\eta+1}\cdot\frac{\eta'}{\eta'+1}=\frac{1}{1+1/\eta}\cdot\frac{1}{1+1/\eta'}.$$
If $0\leq\alpha,\beta,\gamma\leq 1$, then $\rho\geq g(\eta,\eta').$
\end{corollary}
The proof is very similar and thus omitted.
Clearly, $g(\eta,\eta')> f(\eta,\eta')$.
Intuitively, positive $\alpha$, $\beta$, and $\gamma$ suggest positive correlations between $L$, $I$, and $I'$, which is usually the case in real workloads (see Section~\ref{sec:analysis:workloads}).

Based on Theorem~\ref{theorem:lower-bound} and Corollary~\ref{corollary:lower-bound-positive}, we immediately have the following important observation:
\begin{observation}
If $\eta\gg 1$ and $\eta'\gg 1$, we have $f(\eta,\eta')\approx 1$ and $g(\eta,\eta')\approx 1$. As a result, $\rho\approx 1$.
\end{observation}
That is, when both $\eta$ and $\eta'$ are sufficiently large, we should expect very strong correlation between the estimated cost (using Algorithm~\ref{alg:combine-costs}) and the actual cost of a plan.
More generally, it is easy to see both $f(\eta,\eta')$ and $g(\eta,\eta')$ are increasing functions with respect to $\eta$ and $\eta'$.
This implies that we need to increase both $\eta$ and $\eta'$ to improve $\rho$.

\subsubsection{Improve Correlation}

Recall that $\eta=\frac{\sigma_L}{\sigma_I}$ whereas $\eta'=\frac{\lambda\sigma_L}{\sigma_{I'}}$.
$\sigma_L$ and $\sigma_I$ are standard deviation of the \emph{actual} leaf and internal operator CPU time, whereas $\sigma_{I'}$ is the standard deviation of the optimizer's estimated internal cost.
All three are constants for a given workload so we cannot change $\eta$ in Algorithm~\ref{alg:combine-costs}.
On the other hand, $\eta'$ depends on $\lambda$ as well, which depends on our choice of the pivot operator.
Because $\eta'$ increases as $\lambda$ increases, it suggests that we should pick the pivot operator that maximizes $\lambda$, as defined by Equation~\ref{eq:lambda}.
Algorithm~\ref{alg:pivot} presents the details of our selection strategy for the pivot operator based on this idea.

\begin{algorithm}[t]
  \SetAlgoLined
  \KwIn{$\mathcal{O}$, the operators with execution feedback.}
  \KwOut{$o^{\pivot}$, the pivot operator.}
  \SetAlgoLined

  $o^{\pivot}\leftarrow\nil$;
  $\lambda\leftarrow 0$\;
  \ForEach{$o\in\mathcal{O}$}{
    $\lambda_o\leftarrow\frac{\optcost(o)}{\actcost(o)}$\;
    \If{$\lambda_o>\lambda$}{
        $\lambda\leftarrow\lambda_o$\;
        $o^{\pivot}\leftarrow o$\;
    }
  }

  \Return{$o^{\pivot}$}\;
  \caption{Pick the pivot operator.}
\label{alg:pivot}
\end{algorithm}

\subsubsection{A Study of Real Workloads}\label{sec:analysis:workloads}

While the previous analysis characterizes the connection between $\rho$, $\eta$, and $\eta'$, it does not tell us what we should expect in practice.
We thus studied 36 real workloads in the context of index tuning with various physical design (e.g., both row store and column store with necessary indexes) and with at least 10 queries (see Section~\ref{sec:experiments} for details of our settings in index tuning).

\begin{figure}
\centering
\subfigure[Distribution of $\eta$: mean = $369.4$, median = $18.8$.]{ \label{fig:dist:eta}
\includegraphics[width=0.9\columnwidth]{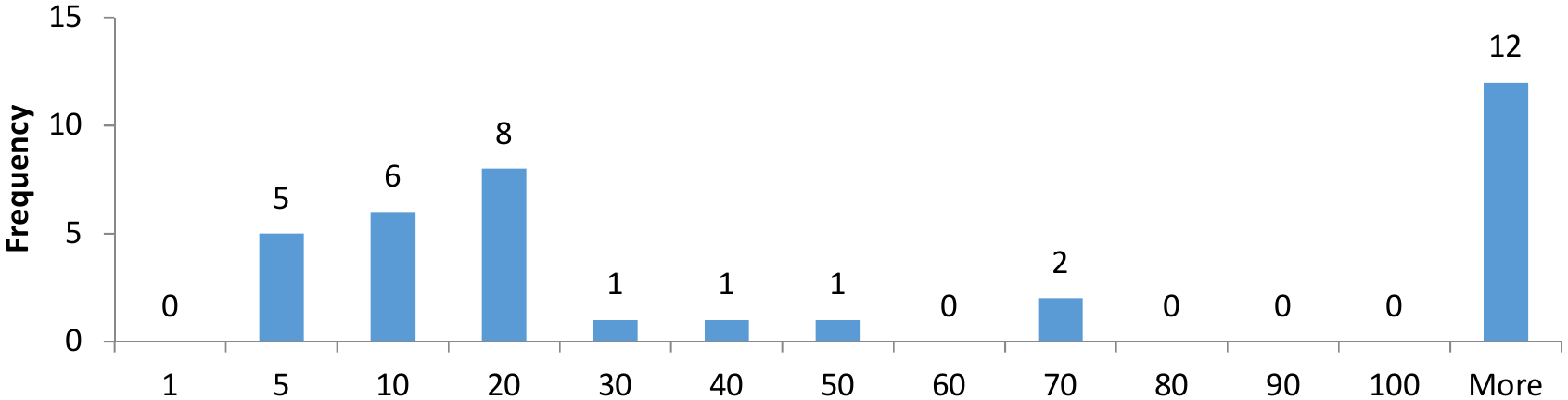}}
\subfigure[Distribution of $\eta'$: mean = $1.2\times 10^7$, median = $6.8\times 10^3$.]{ \label{fig:dist:eta-prim}
\includegraphics[width=0.9\columnwidth]{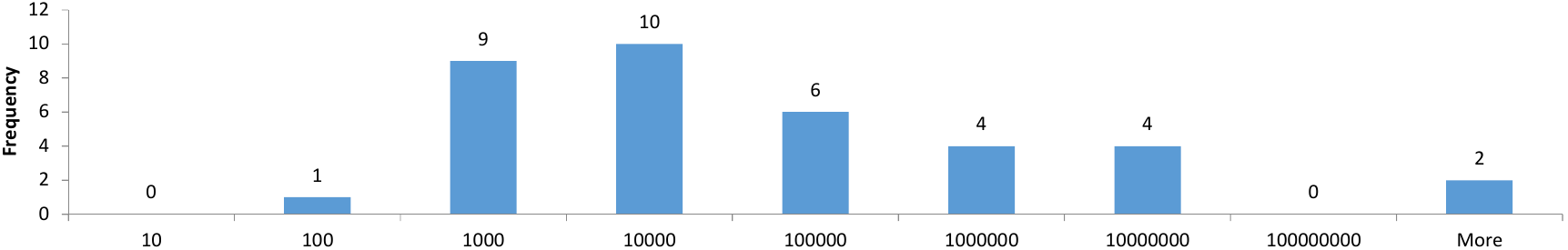}}
\vskip -2ex
\caption{The distributions of $\eta$ and $\eta'$ on real workloads.}
\label{fig:dist:workload}
\vskip -2ex
\end{figure}

Figure~\ref{fig:dist:workload} presents the distributions of $\eta$ and $\eta'$.
We computed $\eta'$ by using Algorithm~\ref{alg:pivot} to pick the pivot operator and therefore $\lambda$.
We observe that $\eta'$ is much larger than $\eta$.
%In fact, the minimum $\eta'$ we observed is $278.1$, which results in $\frac{1}{\eta'}=0.0036$.
Therefore it is safe to ignore the factor $\frac{1}{\eta'}$ in both $f(\eta,\eta')$ and $g(\eta,\eta')$.
Consequently, we have the following approximations:
\begin{observation}
If $\eta'\gg 1$, we have $1/\eta'\approx 0$. As a result, $f(\eta,\eta')\approx \frac{1-1/\eta}{1+1/\eta}$ and $g(\eta,\eta')\approx \frac{1}{1+1/\eta}$.
\end{observation}
Define $f(\eta)=\frac{1-1/\eta}{1+1/\eta}$ and $g(\eta)=\frac{1}{1+1/\eta}$.
Figure~\ref{fig:lower-bound-functions} depicts these two functions with $\eta$ increasing from $1$ to $50$.
We observe that both functions increase quickly when $\eta$ grows.
For example, when $\eta=10$, $f(\eta)=0.81$ and $g(\eta)=0.91$.
When $\eta=18.8$ (i.e., the median we observed on our workloads), we have $f(\eta)=0.90$ and $g(\eta)=0.95$.

\begin{figure}[t]
\centering
    \includegraphics[width=0.8\columnwidth]{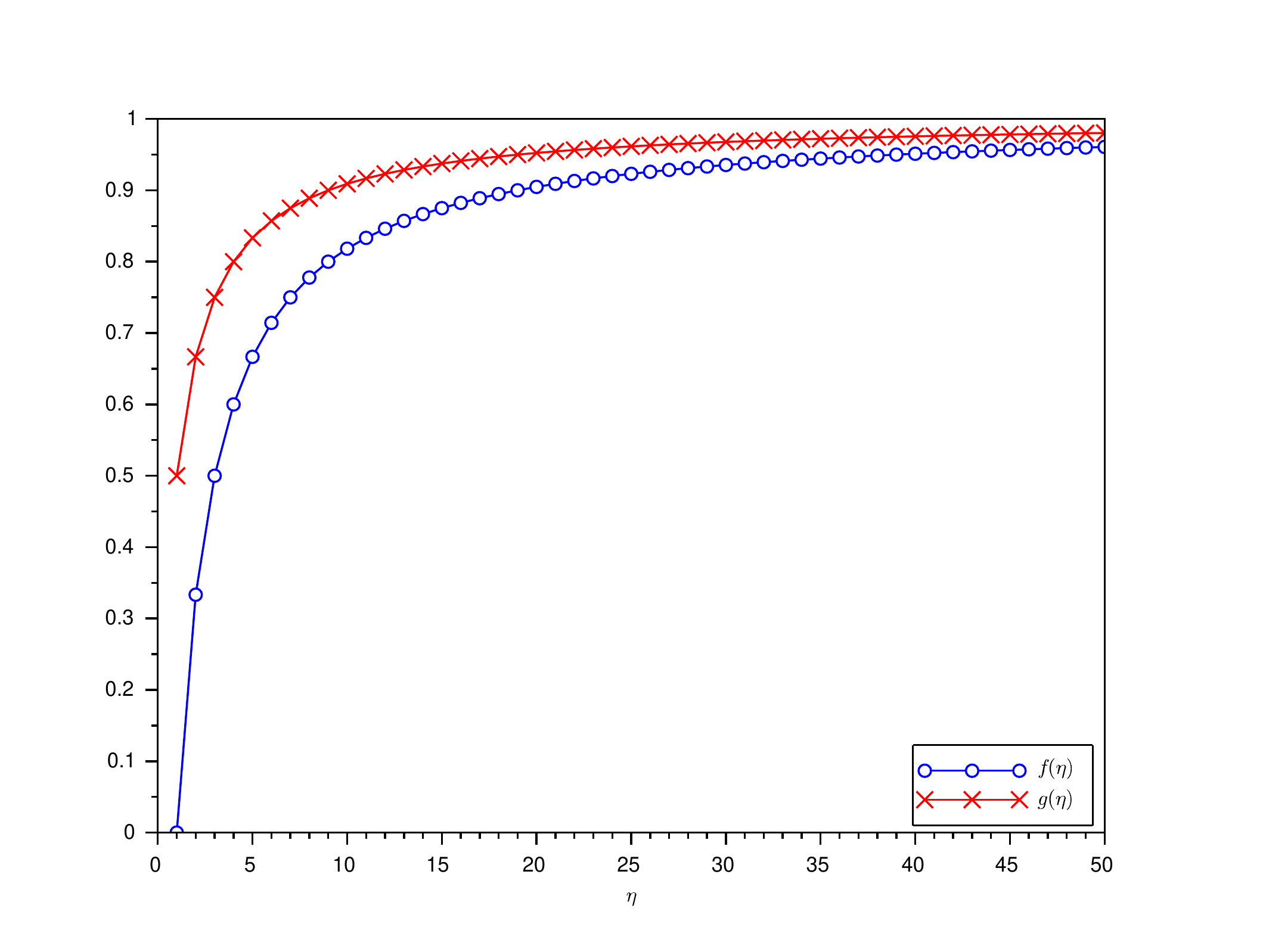}
\vskip -2ex
\caption{Plots of $f(\eta)$ and $g(\eta)$ with the growth of $\eta$.}
\label{fig:lower-bound-functions}
\vskip -2ex
\end{figure}

Under the assumption that $\frac{1}{\eta'}\approx 0$, we can have a more detailed analysis as we will see next.

\subsubsection{The Case When $\eta'$ Is Very Large}

Dividing the numerator and denominator in Equation~\ref{eq:rho} by $\eta'$ and using $\frac{1}{\eta'}\approx 0$, we obtain the following:
\begin{equation}\label{eq:rho-approx}
    \rho\approx\frac{\eta+\alpha}{\sqrt{\eta^2+2\alpha\eta+1}}=\frac{1+\alpha/\eta}{\sqrt{1+2\alpha/\eta+(1/\eta)^2}}.%=\sqrt{\frac{1}{1+\frac{1-\alpha^2}{(\eta+\alpha)^2}}}.
\end{equation}
Again, if $\eta$ is sufficiently large, then $1/\eta\approx 0$ and thus $\rho\approx 1$.
We next view $\rho$ as a function of $\eta$ and $\alpha$. % i.e., $\rho=h(\eta,\alpha)$.
%Figure plots $h(\eta,\alpha)$ with respect to both $\eta$ and $\alpha$.

\begin{lemma}\label{lemma:eta_0}
Assume that Equation~\ref{eq:rho-approx} holds and $\eta+\alpha>0$. For a given $0<\epsilon<1$, there exists some $\eta_0$ s.t. if $\eta>\eta_0$ then $\rho>1-\epsilon$. Specifically,
\begin{equation}\label{eq:eta-0}
  \eta_0=\sqrt{\frac{1-\alpha^2}{1/(1-\epsilon)^2-1}}-\alpha.
\end{equation}
\end{lemma}

\begin{proof}
Using Equation~\ref{eq:rho-approx}, $\rho>1-\epsilon$ implies
$$\eta+\alpha>(1-\epsilon)\cdot\sqrt{\eta^2+2\alpha\eta+1}.$$
Given that $\eta+\alpha>0$, it follows that
$$(\eta+\alpha)^2>(1-\epsilon)^2\cdot\big((\eta+\alpha)^2+(1-\alpha^2)\big).$$
Since $0<1-\epsilon<1$, we have $1/(1-\epsilon)^2>1$. As a result,
$$(\eta+\alpha)^2>\frac{1-\alpha^2}{1/(1-\epsilon)^2-1}.$$
Since $\eta+\alpha>0$, taking the square root completes the proof.
\end{proof}

Lemma~\ref{lemma:eta_0} suggests that there is a minimum $\eta_0$ such that $\rho$ can be sufficiently high as long as $\eta>\eta_0$.
We study two examples below:
\begin{itemize}
    \item If $\epsilon=0.05$, i.e., we want to have $\rho>1-\epsilon=0.95$. As a result,
$\eta_0=\sqrt{(1-\alpha^2)/0.108}-\alpha$.
    \item If $\epsilon=0.01$, i.e., we want to have $\rho>1-\epsilon=0.99$. As a result,
$\eta_0=\sqrt{(1-\alpha^2)/0.0203}-\alpha$.
\end{itemize}

\begin{figure}[t]
\centering
    \includegraphics[width=0.8\columnwidth]{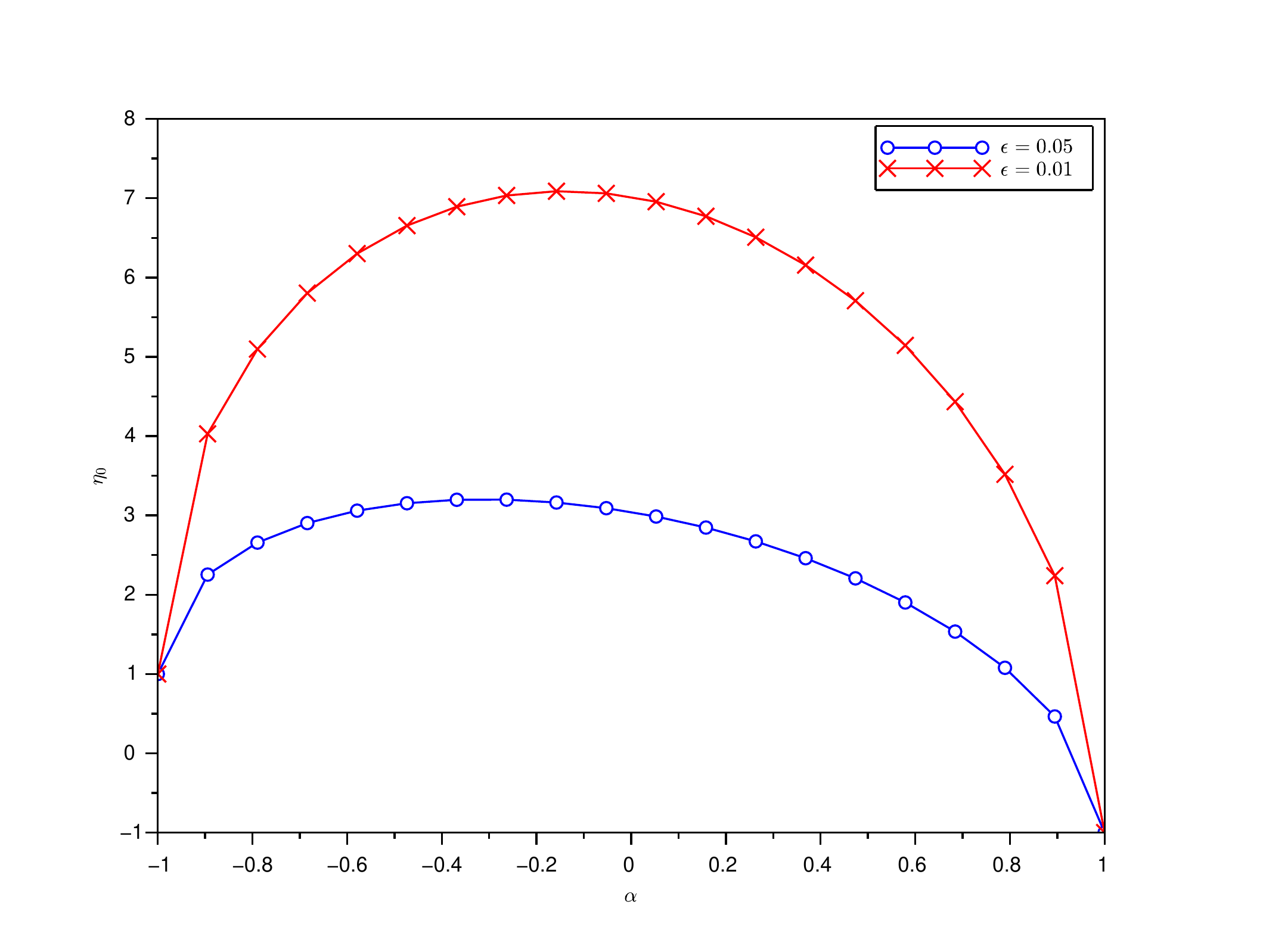}
\vskip -2ex
\caption{Plots of $\eta_0$ as a function of $\alpha$.}
\label{fig:eta-0}
\vskip -2ex
\end{figure}

Figure~\ref{fig:eta-0} plots the $\eta_0$ as a function of $-1\leq\alpha\leq 1$ in the above two cases.
Apparently, $\eta_0$ has a maximum $\eta_0^{\max}$ with $-1\leq\alpha\leq 1$.
As long as $\eta>\eta_0^{\max}$, we will have $\rho>1-\epsilon$ regardless of $\alpha$.
In fact, this is easy to prove using Equation~\ref{eq:eta-0}.
Specifically we have the following theorem.
\begin{theorem}\label{theorem:eta-max}
$\eta_0$ achieves its maximum $\eta_0^{\max}$ when
$$\alpha=-\sqrt{1-(1-\epsilon)^2}.$$
In more detail, we have
\begin{equation}\label{eq:eta-max}
  \eta_0^{\max}=\frac{1}{\sqrt{1-(1-\epsilon)^2}}.
\end{equation}
\end{theorem}

\begin{proof}
We can view $\eta_0$ as a function of $\alpha$, i.e., $\eta_0=\eta_0(\alpha)$.
Define a constant $C=\frac{1}{\sqrt{1/(1-\epsilon)^2-1}}$. By Equation~\ref{eq:eta-0}, we have
$$\eta_0(\alpha)=C\sqrt{1-\alpha^2}-\alpha.$$
Taking derivatives of $\eta_0(\alpha)$, we obtain
$$\eta'_0(\alpha)=-\frac{C\alpha}{\sqrt{1-\alpha^2}}-1,\quad \eta''_0(\alpha)=-\frac{C}{(1-\alpha^2)^{3/2}}.$$
Since $C>0$ and $|\alpha|\leq 1$, we have $\eta''_0(\alpha)<0$.
Therefore, $\eta_0$ achieves its maximum when $\eta'_0(\alpha)=0$.
Letting $\eta'_0(\alpha)=0$ gives
\begin{equation}\label{eq:alpha-max}
  \alpha=-\frac{1}{\sqrt{C^2+1}}=-\sqrt{1-(1-\epsilon)^2}.
\end{equation}
Substituting Equation~\ref{eq:alpha-max} into Equation~\ref{eq:eta-0} gives
\begin{eqnarray*}
\eta_0^{\max}&=&\sqrt{1-(1-\epsilon)^2}+\frac{(1-\epsilon)^2}{\sqrt{1-(1-\epsilon)^2}}\\
&=&\frac{1}{\sqrt{1-(1-\epsilon)^2}}.
\end{eqnarray*}
This completes the proof of the theorem.
\end{proof}

Continuing with the previous examples, by Theorem~\ref{theorem:eta-max} we have
\begin{itemize}
    \item For $\epsilon=0.05$, $\eta_0^{\max}=3.2$ when $\alpha=-0.31$;
    \item For $\epsilon=0.01$, $\eta_0^{\max}=7.1$ when $\alpha=-0.14$.
\end{itemize}
These results can be easily verified in Figure~\ref{fig:eta-0}.
Moreover, by Equation~\ref{eq:alpha-max}, we have $\alpha\to 0$ as $\epsilon\to 0$.
Meanwhile, $\eta_0^{\max}$ increases as $\epsilon$ decreases.
In particular, $\eta_0^{\max}\to\infty$ as $\epsilon\to 0$.
Figure~\ref{fig:eta-0-max} further plots $\eta_0^{\max}$ with respect to $\epsilon$.

\begin{figure}[t]
\centering
    \includegraphics[width=0.8\columnwidth]{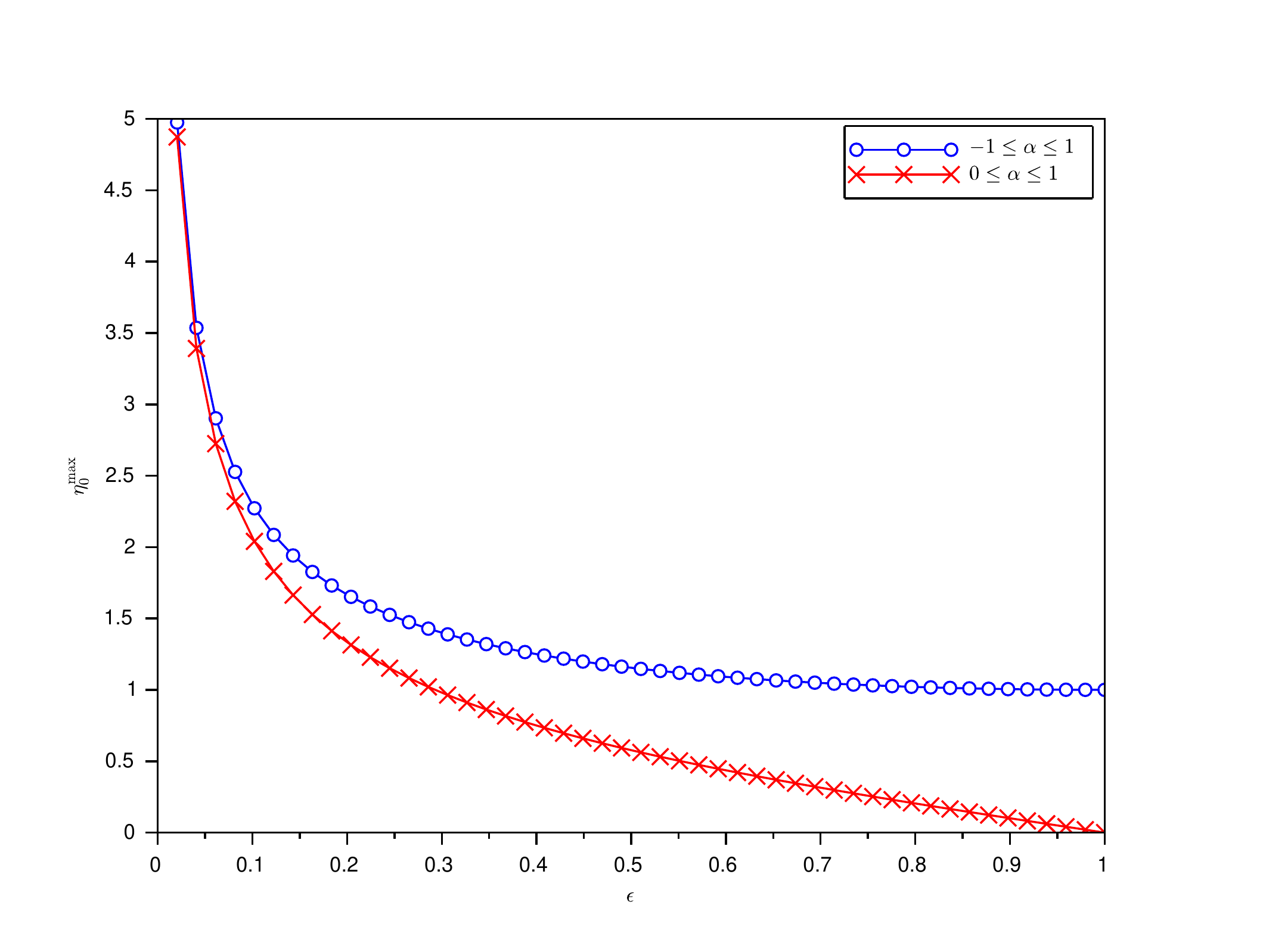}
\vskip -2ex
\caption{Plots of $\eta_0^{\max}$ as a function of $\epsilon$.}
\label{fig:eta-0-max}
\vskip -2ex
\end{figure}

So far we have focused on the general case where $-1\leq\alpha\leq 1$.
In practice, it is reasonable to assume a positive $\alpha$, i.e., $0\leq\alpha\leq 1$.
(For the 36 workloads we studied in Section~\ref{sec:analysis:workloads}, we observed only one workload with a negative $\alpha=-0.09$.)
Therefore, similar to Corollary~\ref{corollary:lower-bound-positive}, we can improve the result given by Theorem~\ref{theorem:eta-max} for the case when $0\leq\alpha\leq 1$.

\begin{corollary}\label{corollary:eta-max-p}
If $0\leq\alpha\leq 1$, $\eta_0$ achieves its maximum $\eta_0^{\max,p}$ when $\alpha=0$ (the superscript $p$ indicates a \emph{positive} $\alpha$):
\begin{equation}\label{eq:eta-max-p}
    \eta_0^{\max,p}=\frac{1-\epsilon}{\sqrt{1-(1-\epsilon)^2}}.
\end{equation}
\end{corollary}

\begin{proof}
By the proof of Theorem~\ref{theorem:eta-max}, we have
$$\eta'_0(\alpha)=-\frac{C\alpha}{\sqrt{1-\alpha^2}}-1, \quad\text{where } C>0.$$
If $\alpha\geq 0$, we have $\eta'_0(\alpha)<0$.
Therefore, $\eta_0(\alpha)$ is a decreasing function of $\alpha$.
As a result, $\eta_0$ achieves its maximum when $\alpha=0$.
Setting $\alpha=0$ in Equation~\ref{eq:eta-0} gives Equation~\ref{eq:eta-max-p}.
\end{proof}

Comparing Equation~\ref{eq:eta-max-p} with Equation~\ref{eq:eta-max} suggests that $\eta_0^{\max,p}<\eta_0^{\max}$.
This means that in the case of a positive $\alpha$, which is the common case in practice, one can have a less stringent requirement on $\eta$ to expect a high $\rho$.
Figure~\ref{fig:eta-0-max} illustrates this difference.
When $\epsilon\to 0$, however, $\eta_0^{\max,p}\to\eta_0^{\max}$.

\subsubsection{Summary and Discussion}

In this section, we started by a formal analysis of Algorithm~\ref{alg:combine-costs}.
We focus on studying the correlation between the cost estimate returned by Algorithm~\ref{alg:combine-costs} and the actual execution cost.
We developed a lower bound for the correlation (Theorem~\ref{theorem:lower-bound} and Corollary~\ref{corollary:lower-bound-positive}) that only depends on $\eta$ and $\eta'$, two quantities determined by workload-level properties.
We then studied many real workloads in the context of index tuning and found that $\eta'$ is typically very large and thus we only need to focus on the part that depends on $\eta$.
We further performed a more detailed analysis in the presence of very large $\eta'$, and found a lower bound for $\eta$ for a given level of correlation we want to achieve (Theorem~\ref{theorem:eta-max} and Corollary~\ref{corollary:eta-max-p}).

\begin{figure}[t]
\centering
\subfigure[Distribution of Pearson CC using optimizer's estimates: mean = $0.54$, median = $0.56$.]{ \label{fig:dist:cost-pearson}
\includegraphics[width=0.9\columnwidth]{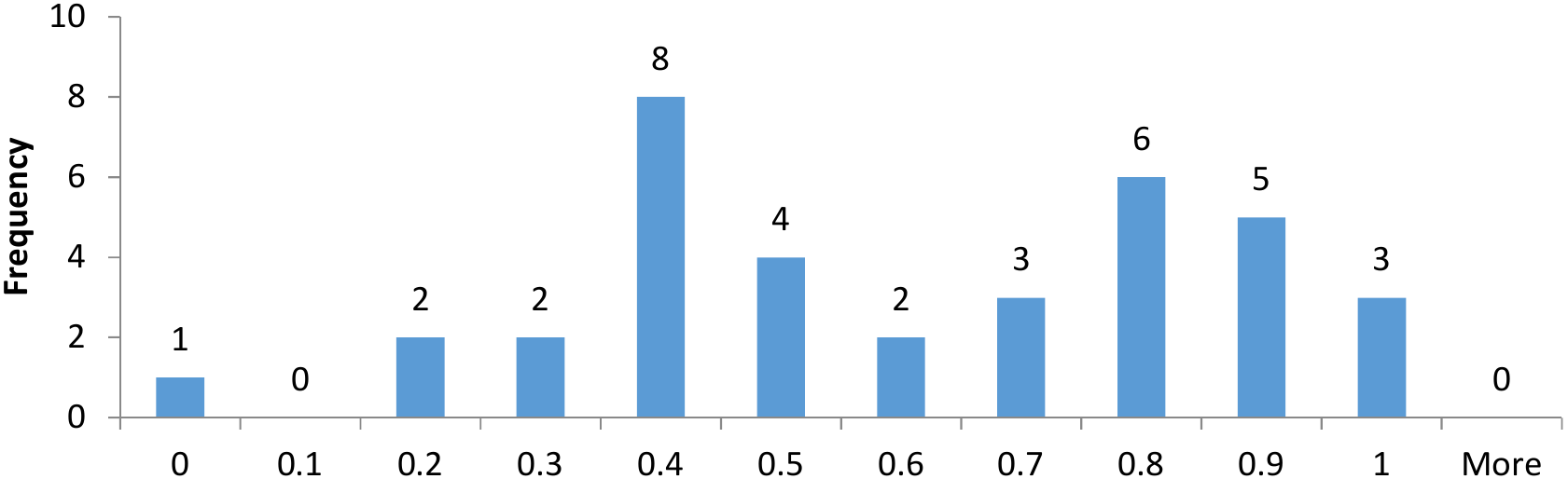}}
\subfigure[Distribution of Pearson CC using Algorithm~\ref{alg:combine-costs} (i.e., $\rho$): mean = $0.81$, median = $0.82$.]{ \label{fig:dist:recost-pearson}
\includegraphics[width=0.9\columnwidth]{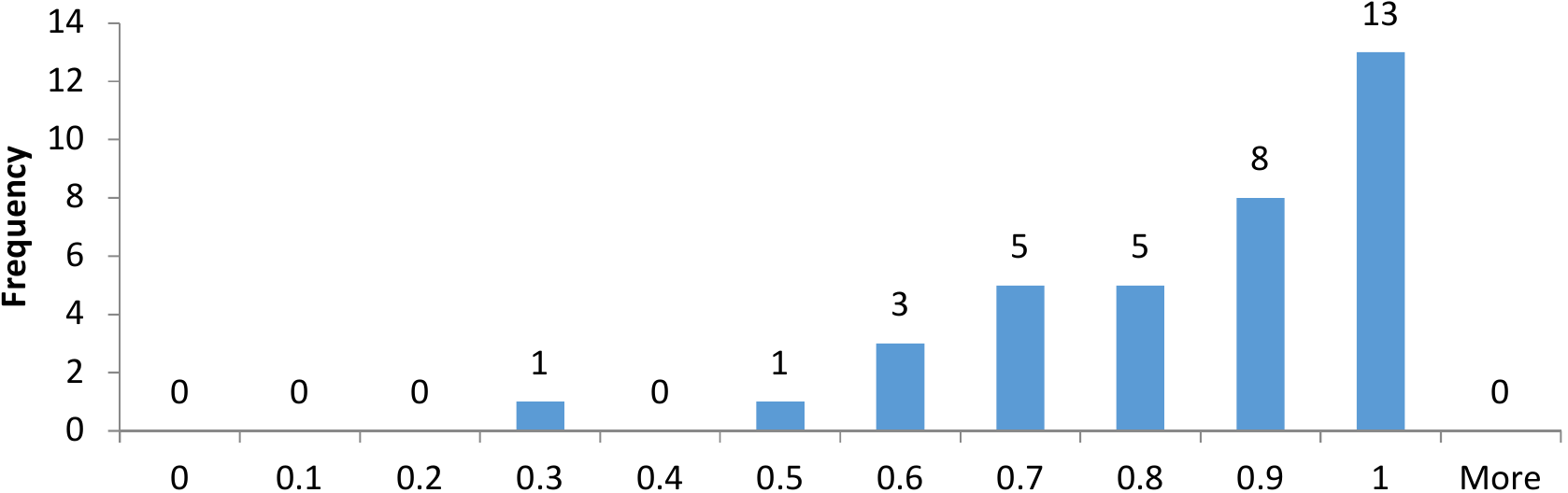}}
\subfigure[Distribution of Spearman CC using optimizer's estimates: mean = $0.53$, median = $0.62$.]{ \label{fig:dist:cost-spearman}
\includegraphics[width=0.9\columnwidth]{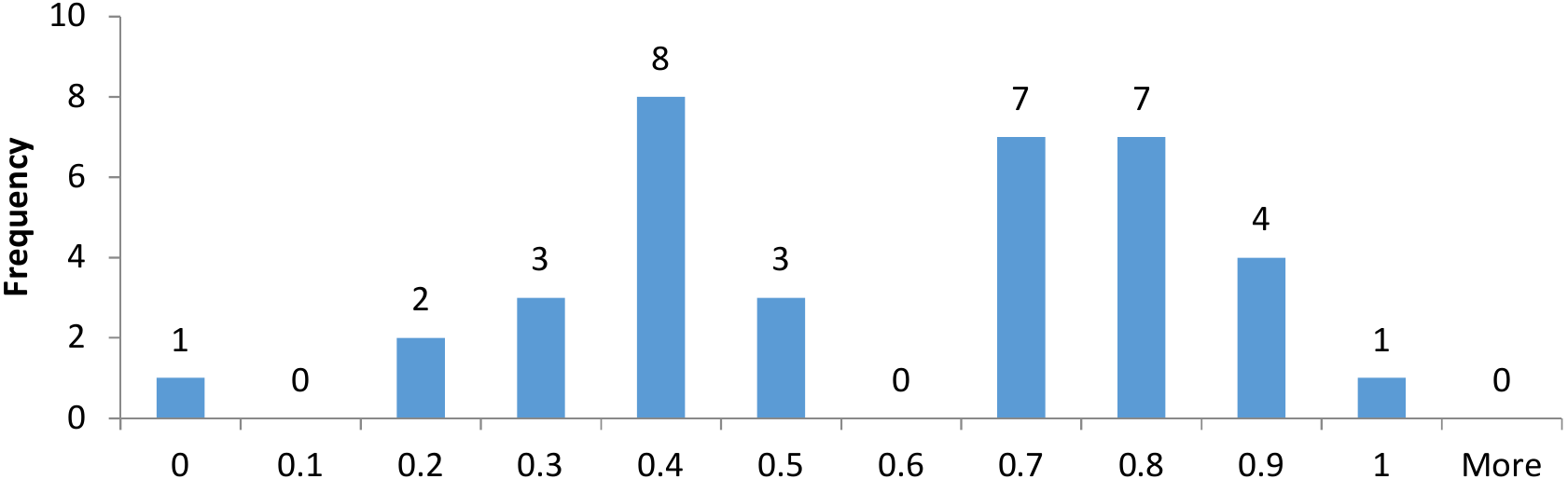}}
\subfigure[Distribution of Spearman CC using Algorithm~\ref{alg:combine-costs}: mean = $0.78$, median = $0.80$.]{ \label{fig:dist:recost-spearman}
\includegraphics[width=0.9\columnwidth]{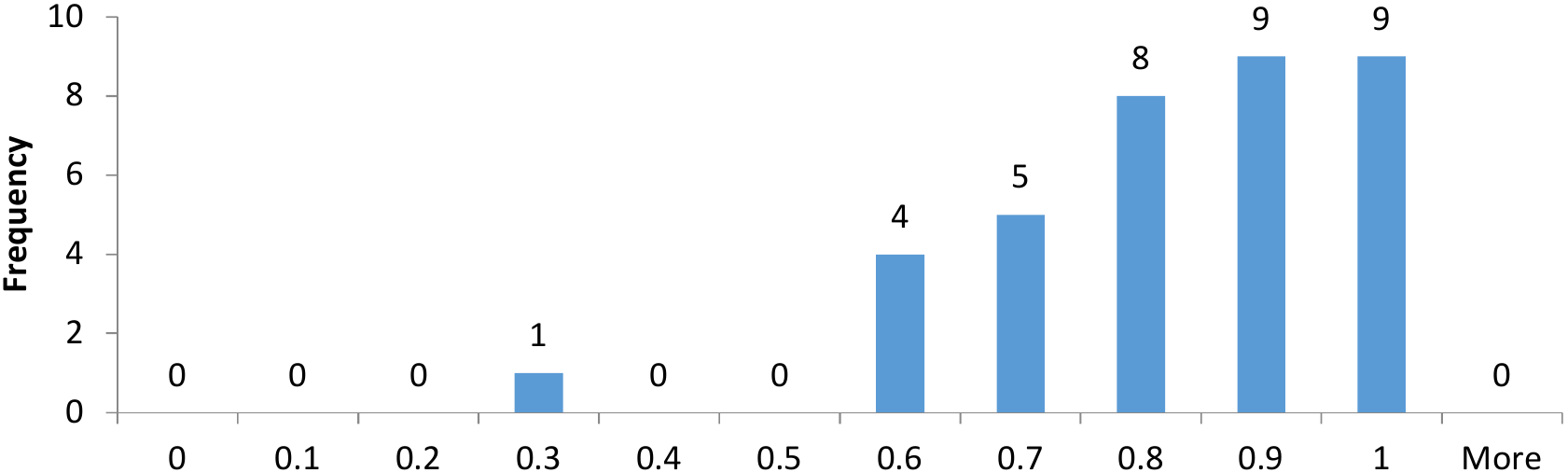}}
\vspace{-1em}
\caption{The distributions of Pearson CC and Spearman CC on real workloads using optimizer's estimates vs. Algorithm~\ref{alg:combine-costs}.}.
\label{fig:dist:cc}
\vspace{-3em}
\end{figure}

As was shown in Figure~\ref{fig:dist:workload}, there is huge variance in the distribution of $\eta$ on real workloads.
Although 25 out of the 36 workloads have $\eta\geq 10$, there are still 11 workloads with relatively small $\eta$.
So a natural question is that how large $\rho$ is over these real workloads.
In Figure~\ref{fig:dist:cc}, we present the distributions of both Pearson CC and Spearman CC on the 36 real workloads.\footnote{Spearman CC is the rank-based version of Pearson CC. Compared to Pearson CC, Spearman CC is more robust when there are outliers, but it ignores the relative differences between costs.}
Comparing with optimizer's cost estimates, the cost estimates returned by Algorithm~\ref{alg:combine-costs} improve the correlation coefficients by $0.55$ to $0.80$ on average.

\vspace{-0.5em}
\paragraph*{Discussion}

Given that $\eta$ plays an important role in determining $\rho$, we further analyze its impact in more detail.
By Equation~\ref{eq:rho-approx}, for a given $\alpha$, we can view $\rho$ as a function of $\eta$, namely, $\rho=\rho(\eta)$.
We have the following simple result.
\begin{lemma}\label{lemma:rho-eta-non-decreasing}
Assume that Equation~\ref{eq:rho-approx} holds. For a given $-1\leq\alpha\leq 1$, $\rho$ is then a non-decreasing function of $\eta$ ( $0\leq\eta\leq\infty$).
\end{lemma}
\begin{proof}
Taking the derivative for $\rho$ (Equation~\ref{eq:rho-approx}), we obtain
$$\rho'(\eta)=\frac{1-\alpha^2}{w^3}, \quad\text{where } w=\sqrt{\eta^2+2\alpha\eta+1}.$$
Since $|\alpha|\leq 1$ and $w>0$, $\rho'(\eta)\geq 0$.
Therefore $\rho$ is a non-decreasing function of $\eta$.
\end{proof}

\begin{theorem}\label{theorem:rho-bounds}
Assuming Equation~\ref{eq:rho-approx} holds, $\alpha\leq\rho\leq 1$.
\end{theorem}

\begin{proof}
By Lemma~\ref{lemma:rho-eta-non-decreasing}, $\rho(\eta)$ is a non-decreasing function of $\eta$.
Given that $0\leq\eta <\infty$, we have $\rho(0)\leq\rho(\eta)\leq\rho(\infty)$.
By Equation~\ref{eq:rho-approx}, $\rho(0)=\alpha$ whereas $\rho(\infty)=1$.
This completes the proof of the theorem.
\end{proof}

In particular, when $|\alpha|<1$, $\rho$ is a strictly increasing function of $\eta$.
When $\alpha=1$, $\rho=1$; when $\alpha=-1$, $\rho=1$ if $\eta\geq 1$, otherwise $\rho=-1$.
In any case Theorem~\ref{theorem:rho-bounds} holds.
Figure~\ref{fig:rho-eta} further plots the two functions $\rho(\eta)=\frac{\eta}{\sqrt{\eta^2+1}}$ when $\alpha=0$ and $\rho(\eta)=\frac{\eta+0.5}{\sqrt{\eta^2+\eta+1}}$ when $\alpha = 0.5$.
It is clear that $\rho\geq\alpha$ in both cases.

\begin{figure}[t]
\centering
    \includegraphics[width=0.8\columnwidth]{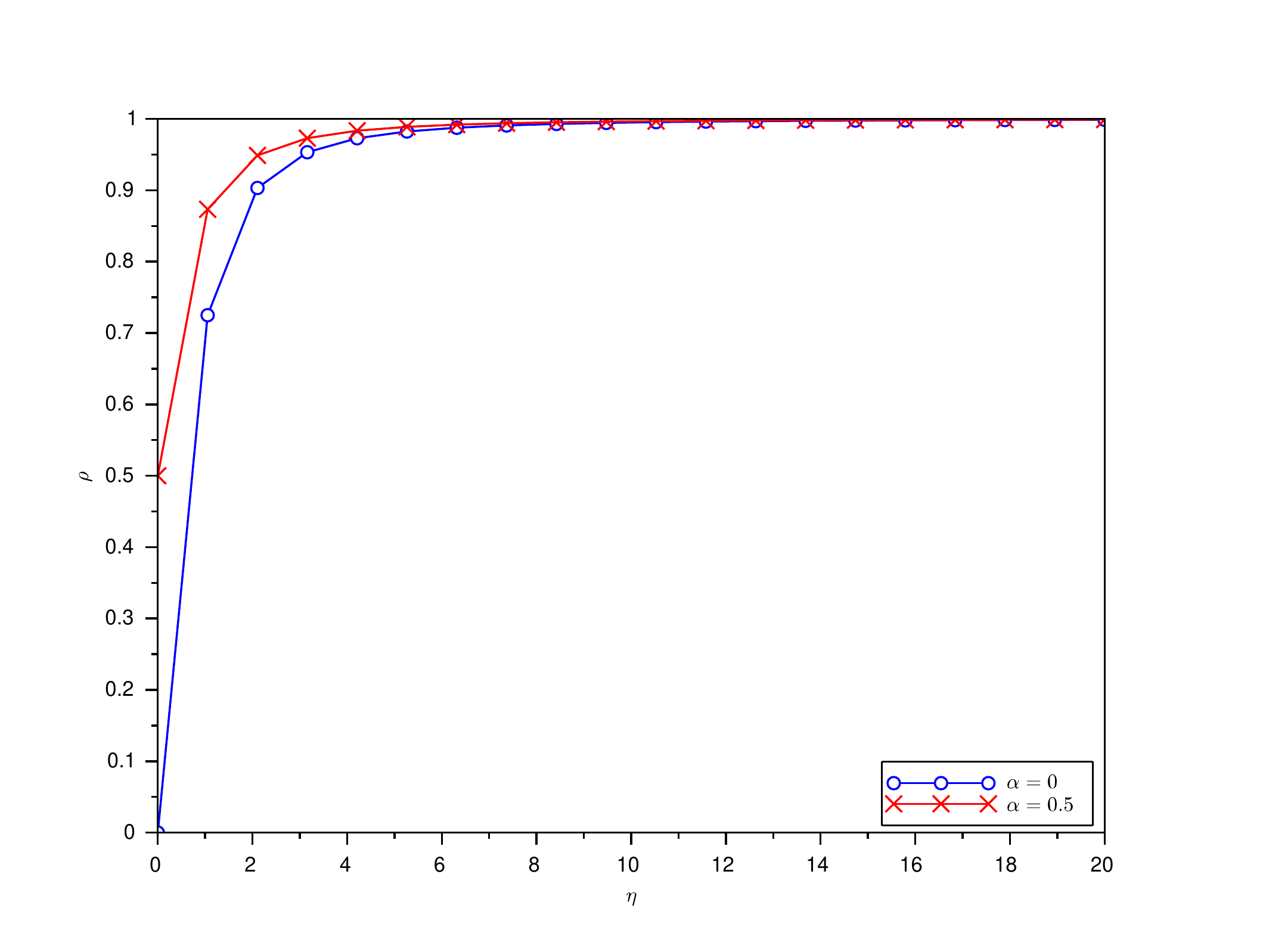}
\vskip -2ex
\caption{Plots of $\rho$ as a function of $\eta$ for a fixed $\alpha$.}
\label{fig:rho-eta}
\vskip -2ex
\end{figure}

In Figure~\ref{fig:summary} we summarize our analysis.
In the presence of a very large $\eta'$, $\rho$ only depends on $\eta$ and $\alpha$.
Given a desired threshold $0<\epsilon<1$, for a given $-1\leq\alpha\leq 1-\epsilon$, along the spectrum $0\leq\eta<\infty$ there exists some $\eta_0$ such that $\rho\geq 1-\epsilon$ when $\eta>\eta_0$.
On the other hand, if $\eta\leq\eta_0$, then a weaker bound for $\rho$ is $\alpha\leq\rho\leq 1-\epsilon$.
Note that the condition $\alpha\leq 1-\epsilon$ is necessary for $\eta_0\geq 0$ (see Equation~\ref{eq:eta-0}).

\begin{figure}[t]
\centering
    \includegraphics[width=0.8\columnwidth]{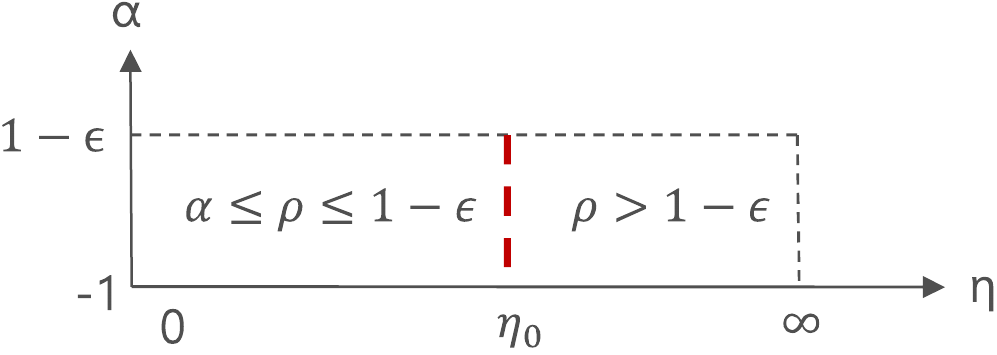}
\vskip -2ex
\caption{Summary of the correlation analysis.}
\label{fig:summary}
\vskip -2ex
\end{figure}

\vspace{-0.5em}
\paragraph*{Remarks}
We have several remarks in order.
First, it is straightforward to extend the analysis to the general case of backbone operators versus the rest (not just leaf operators versus internal ones).
%Second, it is clear that the correlation is not expected to be high when $\eta$ is small.
%This raises the question of how to handle that case.
%We argue that this is unavoidable when there is little execution feedback the set of backbone operators is empty so that $\eta=0$.
%As we accumulate more and more execution feedback, eventually we will arrive at a point where most (if not all) operators have been included into the set of backbone operators (assuming there are no cost modeling errors for backbone operators).
%In that scenario we will have $\eta=\infty$ (if all operators become backbone operators) or a very large $\eta$ (if the set of operators not included is very small).
Second, so far we have focused on the case when $\eta'$ is very large.
This may not hold on certain workloads.
We present more analysis for that situation in Section~\ref{sec:analysis-more}.
Third, so far we have assumed there are no cost modeling errors for backbone operators, which is unlikely the case in practice.
It is straightforward to extend the analysis by considering modeling errors, though the analytic formulas presented in this section will become more complicated.
%For this reason, we omit this extension.

\subsection{The Case When $\eta'$ Is Not Large}\label{sec:analysis-more}

So far we have focused ourselves on the special case when $\eta'$ is very large (more accurately,  $\frac{1}{\eta'}\approx 0$).
One may be also interested in the case when this does not hold.
In the following, we study this case in more detail.

By Equation~\ref{eq:rho}, we can also view $\rho$ as a function of $\eta'$:
\begin{equation}
    \rho=\rho(\eta')=\frac{A(B\eta'+C)}{\sqrt{(\eta')^2+2\beta\eta'+1}},
\end{equation}
where $A=\frac{1}{\sqrt{\eta^2+2\alpha\eta+1}}$, $B=\eta+\alpha$, and $C=\beta\eta+\gamma$.
(As was in Lemma~\ref{lemma:eta_0}, we assume $\eta +\alpha\geq 0$. Note that this automatically holds if $\eta\geq 1$.)
Taking the derivative we obtain
\begin{equation}
    \rho'(\eta')=\frac{ABv^2-u(\eta'+\beta)}{v^3},
\end{equation}
where $u=A(B\eta'+C)$, $v=\sqrt{(\eta')^2+2\beta\eta'+1}$.
Note that the derivative of $v$ satisfies $v'(\eta')=\frac{\eta'+\beta}{v}$.

Now let $\rho'(\eta')=0$. We obtain
\begin{eqnarray}\label{eq:eta-prim-0}
    \eta'_0&=&\frac{\beta C-B}{\beta B-C}\\\nonumber
    &=&\frac{(1-\beta^2)\eta+(\alpha-\beta\gamma)}{\gamma-\alpha\beta}\\\nonumber
    &=&\frac{1-\beta^2}{\gamma-\alpha\beta}\eta+\frac{\alpha-\beta\gamma}{\gamma-\alpha\beta}.
\end{eqnarray}
Using the relation
$$ABv^2|_{\eta'=\eta'_0}=u|_{\eta'=\eta'_0}(\eta'_0+\beta),$$
which gives
$$Bv^2|_{\eta'=\eta'_0}=(B\eta'_0+C)(\eta'_0+\beta),$$
it then follows that
\begin{eqnarray}\label{eq:rho-eta-prim-0}
    \rho(\eta'_0)&=&\frac{A(B\eta'_0+C)}{v|_{\eta'=\eta'_0}}\\\nonumber
    &=&\frac{ABv|_{\eta'=\eta'_0}}{\eta'_0+\beta}\\\nonumber
    &=&\frac{AB\sqrt{(\eta'_0)^2+2\beta\eta'_0+1}}{\eta'_0+\beta}\\\nonumber
    &=&AB\cdot\sqrt{1+\frac{1-\beta^2}{(\eta'_0+\beta)^2}}\\\nonumber
    &=&\sqrt{\frac{(\eta+\alpha)^2+\frac{(\gamma-\alpha\beta)^2}{1-\beta^2}}{(\eta+\alpha)^2+(1-\alpha^2)}}.
\end{eqnarray}

Furthermore, we have
\begin{equation*}
    \rho''(\eta')=\frac{AB(\beta-C)v^3-3v(\eta'+\beta)[ABv^2-u(\eta'+\beta)]}{v^6}.
\end{equation*}
Again, using the relation
$$ABv^2|_{\eta'=\eta'_0}=u|_{\eta'=\eta'_0}(\eta'_0+\beta),$$
it follows that
\begin{equation}
    \rho''(\eta'_0)=\frac{AB(\beta-C)}{v^3|_{\eta'=\eta'_0}}=\frac{AB[(1-\eta)\beta-\gamma]}{v^3|_{\eta'=\eta'_0}}.
\end{equation}
Assume $\beta >0$.
Therefore, if $\eta>1-\frac{\gamma}{\beta}$, we have $\rho''(\eta'_0)<0$ and thus $\rho(\eta')$ attains its maximum at $\eta'_0$.
On the other hand, if $\eta<1-\frac{\gamma}{\beta}$, $\rho''(\eta'_0)>0$ and thus $\rho(\eta')$ attains its minimum at $\eta'_0$.

Moreover, if $\gamma >\beta$, then $1-\frac{\gamma}{\beta}<0$.
Thus $\eta>1-\frac{\gamma}{\beta}$ always holds and $\rho(\eta')$ attains its maximum at $\eta'_0$.
On the other hand, if $0<\gamma\leq\beta$, then $0\leq 1-\frac{\gamma}{\beta}<1$.
If $\eta\geq 1$ then $\rho(\eta')$ still attains its maximum at $\eta'_0$.
Otherwise we need to compare $\eta$ and $1-\frac{\gamma}{\beta}$.
Hence, we have proved the following result:
\begin{theorem}\label{theorem:rho-upper-lower-bounds}
    If $\eta\geq 1$ and $0<\beta, \gamma<1$, then $\rho(\eta')$ attains its maximum $\rho_{\max}=\rho(\eta'_0)$ (Equation~\ref{eq:rho-eta-prim-0}) at $\eta'_0$, and attains its minimum $\rho_{\min}$ at either
    $$\rho(0)=AC=\frac{\beta\eta+\gamma}{\sqrt{\eta^2+2\alpha\eta+1}}$$ or
    $$\rho(\infty)=AB=\frac{\eta+\alpha}{\sqrt{\eta^2+2\alpha\eta+1}}.$$ We have $\rho_{\min}\leq\rho\leq\rho_{\max}$.
\end{theorem}

As an exercise, let us verify $\rho(\eta'_0)\geq\rho(0)$ and $\rho(\eta'_0)\geq\rho(\infty)$ when $\eta\geq 1$.
Clearly $\rho(\eta'_0)\geq\rho(\infty)$.
We now prove that $\rho(\eta'_0)\geq\rho(0)$ under the condition of Theorem~\ref{theorem:rho-upper-lower-bounds}.
\begin{corollary}
If $\eta\geq 1$ and $0<\beta, \gamma<1$, $\rho(\eta'_0)\geq\rho(0)$.
\end{corollary}
\begin{proof}
First notice that we must have $\eta'_0\geq 0$. By Equation~\ref{eq:eta-prim-0},
$$\frac{1-\beta^2}{\gamma-\alpha\beta}\eta+\frac{\alpha-\beta\gamma}{\gamma-\alpha\beta}\geq 0.$$
If $\gamma <\alpha\beta$, then we have
$$(1-\beta^2)\eta + (\alpha-\beta\gamma)\leq 0,$$
Note that, under the condition $0<\beta, \gamma<1$, we must have $0<\gamma <\alpha\beta$.
As a result, $\alpha >0$ and
$$\eta\leq\frac{\beta\gamma-\alpha}{1-\beta^2}\leq\frac{\beta\cdot\alpha\beta-\alpha}{1-\beta^2}=\frac{\alpha(\beta^2-1)}{1-\beta^2}=-\alpha.$$
This is impossible given that $\eta\geq 1$ and $\alpha >0$.
Consequently, we must have $\gamma\geq\alpha\beta$. Hence,
$$(1-\beta^2)\eta + (\alpha-\beta\gamma)\geq 0,$$
which gives
\begin{equation}\label{eq:eta-condition}
    \eta\geq\frac{\beta\gamma-\alpha}{1-\beta^2}.
\end{equation}
Now define
\begin{eqnarray*}\label{eq:Delta}
    \Delta&=&\rho^2(\eta'_0)-\rho^2(0)\\\nonumber
    &=&A^2\big((\eta+\alpha)^2+\frac{(\gamma-\alpha\beta)^2}{1-\beta^2}-(\beta\eta+\gamma)^2\big)\\\nonumber
    &=&\frac{A^2}{1-\beta^2}\big[(1-\beta^2)\big((\eta+\alpha)^2-(\beta\eta+\gamma)^2\big)+(\gamma-\alpha\beta)^2\big].
\end{eqnarray*}
Next, let us consider
\begin{eqnarray}\label{eq:delta}
\delta &=& (\eta+\alpha)^2-(\beta\eta+\gamma)^2\\\nonumber
&=&(1-\beta^2)\eta^2+2(\alpha-\beta\gamma)\eta + \alpha^2-\gamma^2.
\end{eqnarray}
\textbf{(Case 1)} Suppose that $\beta\gamma-\alpha\geq 0$.
Plugging Equation~\ref{eq:eta-condition} into Equation~\ref{eq:delta}, it follows that
\begin{eqnarray*}
\delta &\geq& (1-\beta^2)\cdot\frac{(\beta\gamma-\alpha)^2}{(1-\beta^2)^2}+2(\alpha-\beta\gamma)\cdot\frac{\beta\gamma-\alpha}{1-\beta^2} + \alpha^2-\gamma^2\\\nonumber
&=&\frac{1}{1-\beta^2}\big[-(\alpha-\beta\gamma)^2+(\alpha^2-\gamma^2)(1-\beta^2)\big]\\
&=&\frac{1}{1-\beta^2}(2\alpha\beta\gamma-\gamma^2-\alpha^2\beta^2)\\
&=&-\frac{1}{1-\beta^2}(\gamma-\alpha\beta)^2.
\end{eqnarray*}
Therefore, we can obtain
\begin{eqnarray*}
\Delta &\geq & \frac{A^2}{1-\beta^2}\big[(1-\beta^2)\big(-\frac{1}{1-\beta^2}(\gamma-\alpha\beta)^2\big)+(\gamma-\alpha\beta)^2\big]\\
&=& \frac{A^2}{1-\beta^2}\big[-(\gamma-\alpha\beta)^2+(\gamma-\alpha\beta)^2\big]\\
&=& 0.
\end{eqnarray*}
\textbf{(Case 2)} On the other hand, suppose that $\beta\gamma-\alpha < 0$.
This implies $\alpha >\beta\gamma$. Consequently, by Equation~\ref{eq:delta} and the condition $\eta\geq 1$,
%we obtain
\begin{eqnarray*}
\delta &\geq & (1-\beta^2)\eta^2+\gamma^2\beta^2-\gamma^2\\
&\geq & (1-\beta^2) + \gamma^2(\beta^2-1)\\
&=&(1-\beta^2)(1-\gamma^2)\\
&\geq & 0.
\end{eqnarray*}
As a result, it follows that
\begin{equation*}
\Delta \geq \frac{A^2}{1-\beta^2}(\gamma-\alpha\beta)^2\geq 0.
\end{equation*}
Combining Case 1 and Case 2, it always holds that $\Delta\geq 0$.
Therefore $\rho(\eta'_0)\geq \rho(0)$. This completes the proof.
\end{proof}

\section{Case Study: Index Tuning}\label{sec:experiments}

We now present a case study that applies the framework in Algorithm~\ref{alg:combine-costs} to index tuning.
We implemented a prototype system piggybacked on the index tuning framework presented in~\cite{ChaudhuriN97}, which has been incorporated into Microsoft SQL Server~\cite{AgrawalCKMNS05,ChaudhuriN07}.

\subsection{System Architecture}\label{sec:experiments:architecture}

There has been extensive research in the area of index tuning over the past two decades~\cite{ChaudhuriN07}.
Major commercial database systems, including Oracle~\cite{DagevilleDDYZZ04}, IBM DB2~\cite{ValentinZZLS00}, and Microsoft SQL Server~\cite{ChaudhuriN97}, are all equipped with index tuning tools.
At the highest level, index tuning is very similar to query optimization.
Existing tools adopt a cost-based approach that picks an index configuration from a number of candidates that results in the minimum optimizer's estimated cost for a given workload consisting of multiple queries~\cite{ChaudhuriN97}.
A basic step here is then to estimate cost for a query over a candidate index configuration.
While different systems may have different implementations, they essentially rely on the so-called ``what-if'' utility that allows the optimizer to generate query plans and estimate their costs for a given query using ``hypothetical'' or ``virtual'' indexes that are not actually materialized~\cite{ChaudhuriN98}.

\begin{figure}[t]
\centering
    \includegraphics[width=0.9\columnwidth]{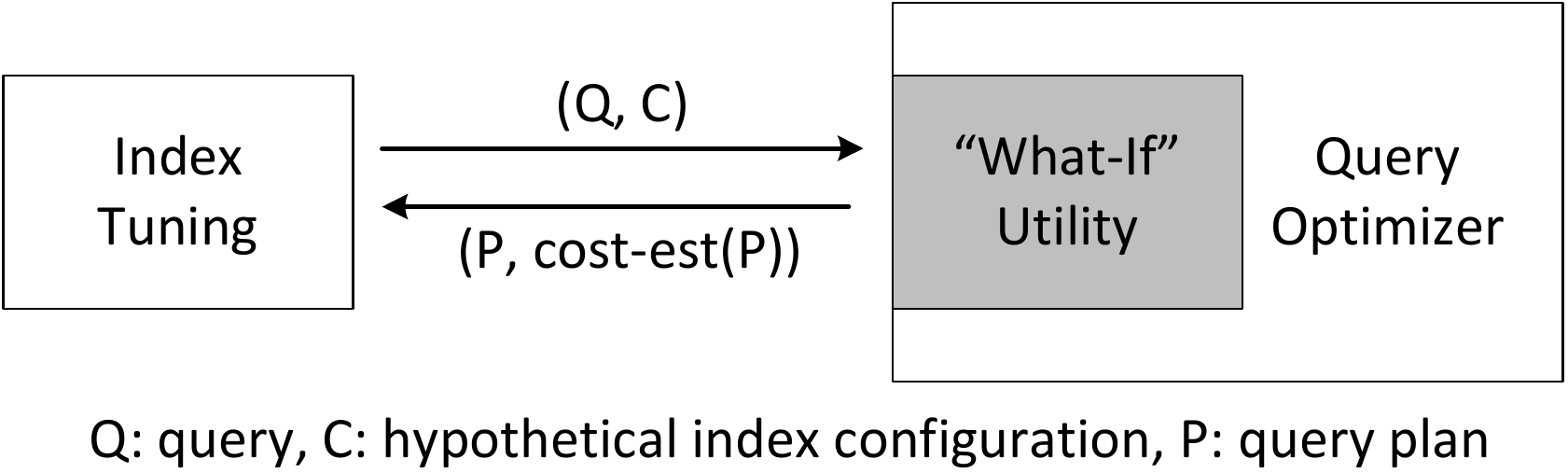}
\vskip -2ex
\caption{Architecture of index tuning using ``what-if'' utility.}
\label{fig:what-if-architecture}
\vskip -2ex
\end{figure}

Figure~\ref{fig:what-if-architecture} outlines this architecture.
The index tuning component sends a query $Q$ with the description of a candidate index configuration $C$ to the optimizer.
The ``what-if'' utility simulates $C$ by generating hypothetical indexes $C$ contains.
That is, it generates all metadata and statistics information about these indexes and makes them visible to the optimizer (only in this particular session).
The key observation here is that query optimization does not require physical persistence of these indexes: The optimizer only needs metadata and statistics to estimate costs for query plans that use these indexes.
The best plan $P$ chosen by the optimizer under the configuration $C$, along with optimizer's cost estimate for $P$, is returned to the index tuning component.

\begin{figure}[t]
\centering
    \includegraphics[width=0.9\columnwidth]{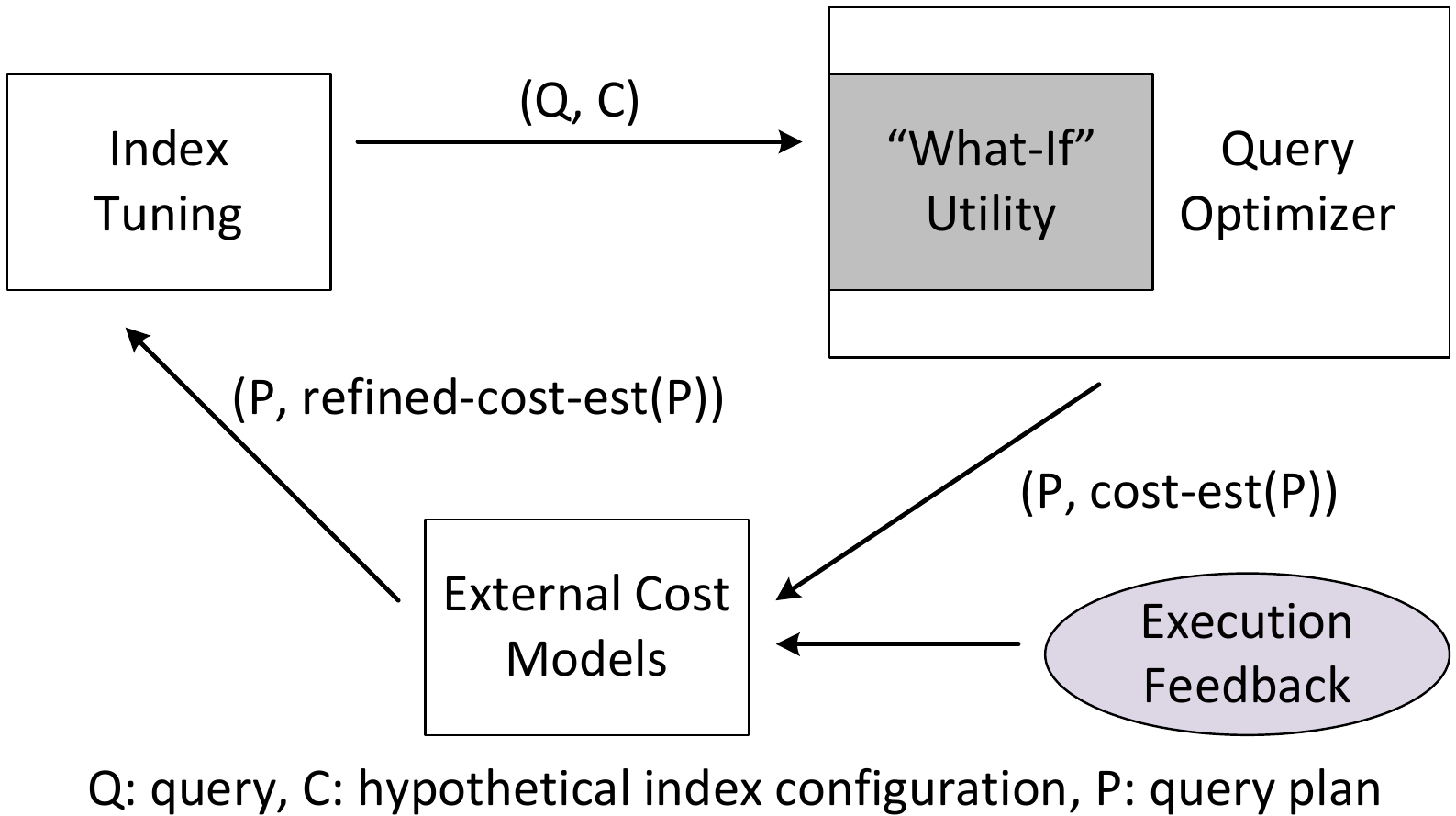}
\vspace{-1em}
\caption{Architecture of index tuning using feedback.}
\label{fig:feedback-architecture}
\vspace{-1em}
\end{figure}

Our idea of incorporating feedback into the architecture of index tuning is simple.
As illustrated in Figure~\ref{fig:feedback-architecture}, after the optimizer returns the best plan $P$, we refine its cost estimate using external cost models built on top of available execution feedback (i.e., Algorithm~\ref{alg:combine-costs}).
We then send the refined cost estimate, instead of the original cost estimate from the optimizer, to the index tuning component.
Note that our approach is \emph{passive} rather than \emph{proactive}: We do not use external cost models inside query optimizer to affect its plan choice.
As a result, if the query optimizer ends up choosing a poor execution plan because of bad cost estimates, we cannot bail it out.
However, by refining cost estimates afterwards, we increase the chance of \emph{detecting} such bad plans and therefore avoiding corresponding disastrous index configurations that may lead to serious query performance regression.

\subsection{Experimental Settings}

The effectiveness of Algorithm~\ref{alg:combine-costs} relies on the following factors:
(1) the backbone operators $\mathcal{O}$; (2) the operator-level cost models $\mathcal{M}$; and (3) the execution feedback $\mathcal{F}$.
For (1), as we have discussed in Section~\ref{sec:framework:implementation}, we use leaf operators as backbone operators;
For (2), we use the operator-level modeling approach presented in~\cite{Li12Robust}, as it represents the state of the art to the best of our knowledge;
For (3), we assume sufficient amount of execution feedback is available for leaf operators.

We used both synthetic and real database workloads in our evaluation.
For synthetic data, we used both TPC-H and TPC-DS at the scale of 10GB; for real data, we used three customer databases Real-1, Real-2, and Real-3.
Table~\ref{tab:databases} presents the details of the workloads we used and their characteristics.
$\eta'$ is very large over all of these workloads.
We focused on single-query workloads and conducted experiments under a workstation configured with two Intel E5-2650L CPUs and 192GB main memory.
We ran Microsoft SQL Server 2017 under 64-bit Windows Server 2012.

\begin{table}%[!htb]
\centering
\begin{tabular}{|l|r|r|r|r|}
\hline
Name & Size & \#Queries & $\eta$ & $\eta'$\\
\hline
\hline
TPC-DS & 10GB & 99 & 56.6 & $3.1\times 10^4$\\
TPC-H & 10GB & 22 & 51.2 & $4.8\times 10^3$\\
\hline
\hline
Real-1 & 40GB & 25 & 1.7 & $4.5\times 10^3$\\
Real-2 & 60GB & 12 & 429.9 & $5.5\times 10^5$\\
Real-3 & 100GB & 20 & 217.9 & $2.4\times 10^6$\\
\hline
\end{tabular}
\caption{Workloads used in experimental evaluation. $\eta$ and $\eta'$ are averaged over all queries in the workload.}
\label{tab:databases}
\vskip -2ex
\end{table}

\subsection{Evaluation Results}

\vspace{-0.5em}
\paragraph*{Initial Configuration}
Since index tuning needs to start from an initial configuration, we generated various initial configurations for our experiments in the following way.
For each query $q$ in the workload, we generated different index configurations by limiting the number of indexes recommended by the index tuning tool (without using execution feedback).
Specifically, we keep asking the index tuning tool to return the next index until it runs out of recommendations.
Suppose that the indexes recommended subsequently are $i_1, ..., i_n$.
We then have $n$ configurations $I_1=\{i_1\}$, $I_2=\{i_1, i_2\}$, ..., and $I_n=\{i_1, i_2, ..., i_n\}$.
We used each of these $n$ configurations as a different initial configuration.

\vspace{-0.5em}
\paragraph*{Execution Feedback}

We generate execution feedback in the following manner.
For each initial configuration, we run the query and collect its execution time.
For this purpose, we enable the ``statistics XML'' utility provided by Microsoft SQL Server to track operator-level execution information.
We then randomly pick up to five executed query plans into the execution feedback repository $\mathcal{F}$.

\vspace{-0.5em}
\paragraph*{Performance Metrics}
We evaluate both the effectiveness of Algorithm~\ref{alg:combine-costs} and the overall improvement of index tuning when execution feedback is utilized, with the following metrics:
\begin{enumerate}
    \item \textbf{(Effectiveness of Algorithm~\ref{alg:combine-costs})} Following our discussion in Section~\ref{sec:analysis:metrics}, we use the Pearson and Spearman correlation coefficients as our performance metrics.
    \item \textbf{(Overall improvement)} We measure the \emph{relative} improvement of the index configuration $I^{\new}$ returned by index tuning over the original index configuration $I^{\old}$, defined as follows.
    Let $c(q, I)$ and $a(q, I)$ be the estimated and actual execution costs of $q$ over a configuration $I$, respectively.
    We define the \emph{estimated improvement} of $I^{\new}$ over $I^{\old}$ as
    \begin{eqnarray*}
        c(I^{\old}, I^{\new})&=&\big(c(q, I^{\old}) - c(q, I^{\new})\big) / c(q, I^{\old})\\
        &=& 1 - c(q, I^{\new}) / c(q, I^{\old}).
    \end{eqnarray*}
    We also define the \emph{actual improvement} of $I^{\new}$ over $I^{\old}$ as
    \begin{eqnarray*}
        a(I^{\old}, I^{\new})&=&\big(a(q, I^{\old}) - a(q, I^{\new})\big) / a(q, I^{\old})\\
        &=& 1 - a(q, I^{\new}) / a(q, I^{\old}).
    \end{eqnarray*}
    We use the actual improvement as our metric, whereas the estimated improvement is useful for controlling the index tuning component, as we will see.
\end{enumerate}

\vspace{-1em}
\paragraph*{Results}

\begin{figure}
\centering
    \includegraphics[angle=90, width=\columnwidth]{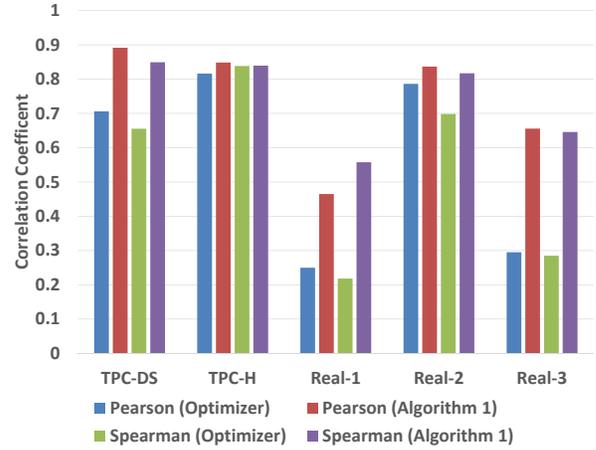}
\vspace{-3em}
\caption{Comparison of correlation coefficients using cost estimates from Algorithm~\ref{alg:combine-costs} over using optimizer's estimates.}
\label{fig:cc}
\vspace{-1em}
\end{figure}

Figure~\ref{fig:cc} presents the correlation coefficients between estimated costs and actual CPU times.
We compare the correlation coefficients using cost estimates produced by Algorithm~\ref{alg:combine-costs} against ones using optimizer's estimates.
We observe significant improvement over four of the five workloads.
This implies that the cost estimates from Algorithm~\ref{alg:combine-costs} are considerably better than optimizer's original cost estimates.

\begin{figure}[t]
\centering
\subfigure[Index tuning without feedback.]{ \label{fig:dist:improvement-no-feedback-0}
\includegraphics[width=0.9\columnwidth]{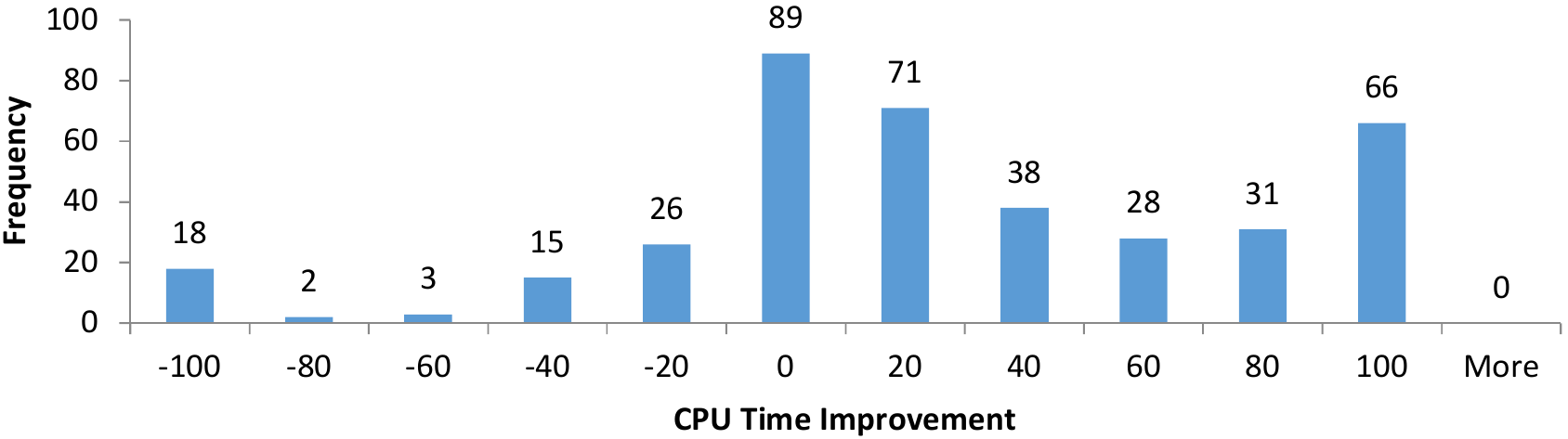}}
\subfigure[Index tuning with feedback.]{ \label{fig:dist:improvement-feedback-0}
\includegraphics[width=0.9\columnwidth]{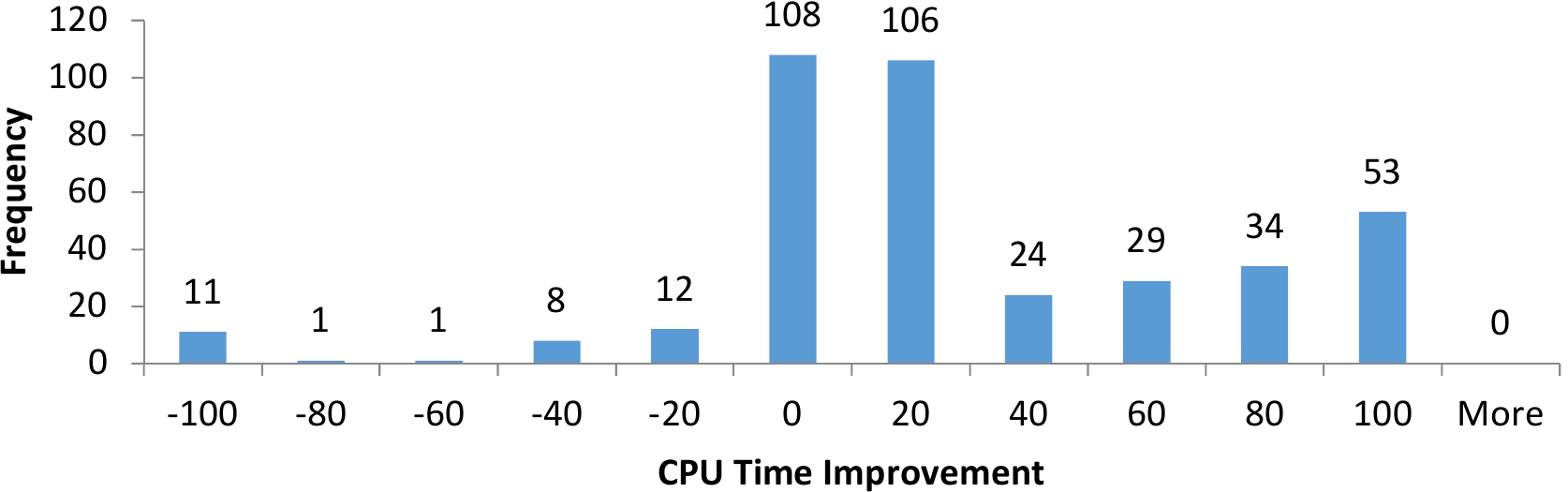}}
\vspace{-1em}
\caption{The distributions of CPU time improvement over TPC-DS queries (estimated improvement threshold $\tau = 0$).}
\label{fig:dist:cpu-time-TPC-DS-0}
\vspace{-1em}
\end{figure}

\begin{figure}[t]
\centering
\subfigure[Index tuning without feedback.]{ \label{fig:dist:improvement-no-feedback-10}
\includegraphics[width=0.9\columnwidth]{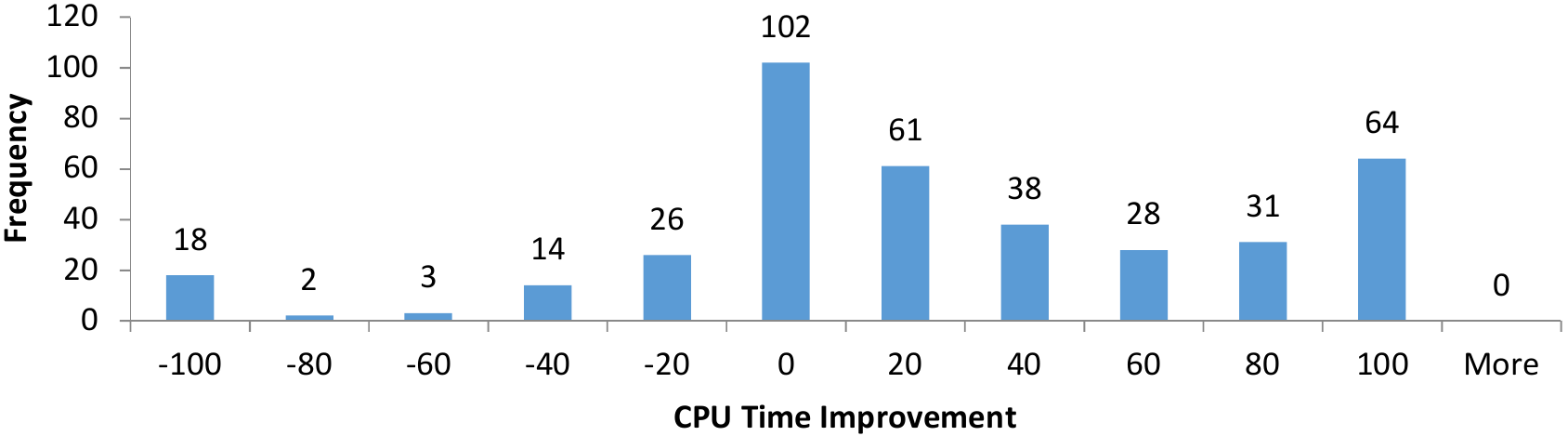}}
\subfigure[Index tuning with feedback.]{ \label{fig:dist:improvement-feedback-10}
\includegraphics[width=0.9\columnwidth]{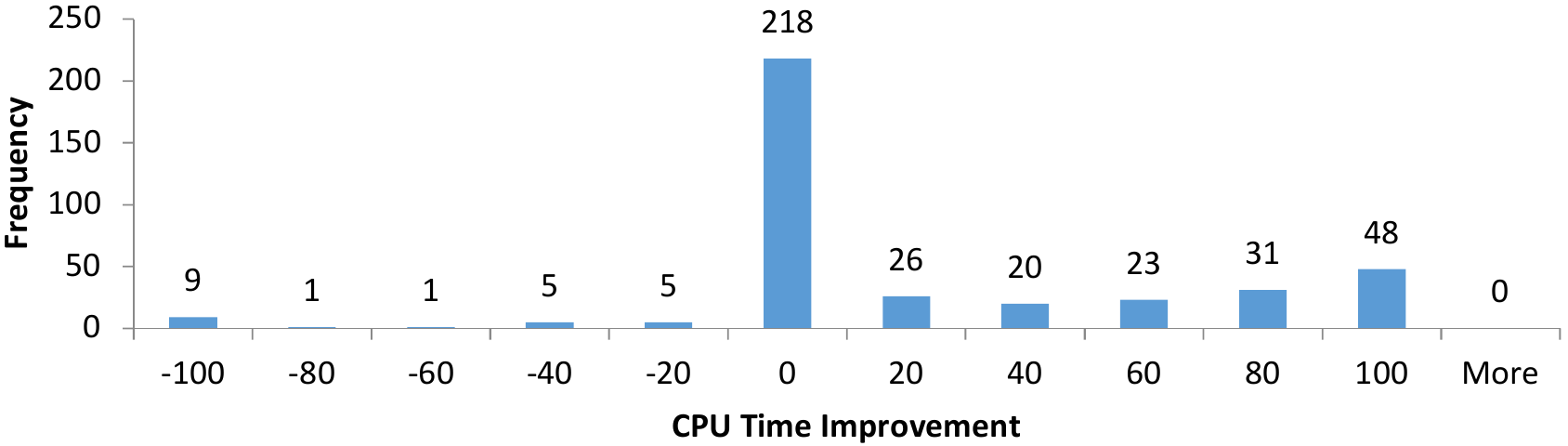}}
\vskip -2ex
\caption{The distributions of CPU time improvement over TPC-DS queries (estimated improvement threshold $\tau = 0.1$).}
\label{fig:dist:cpu-time-TPC-DS-10}
\vskip -3ex
\end{figure}

\begin{figure}[t]
\centering
\subfigure[Index tuning without feedback.]{ \label{fig:dist:improvement-no-feedback-20}
\includegraphics[width=0.9\columnwidth]{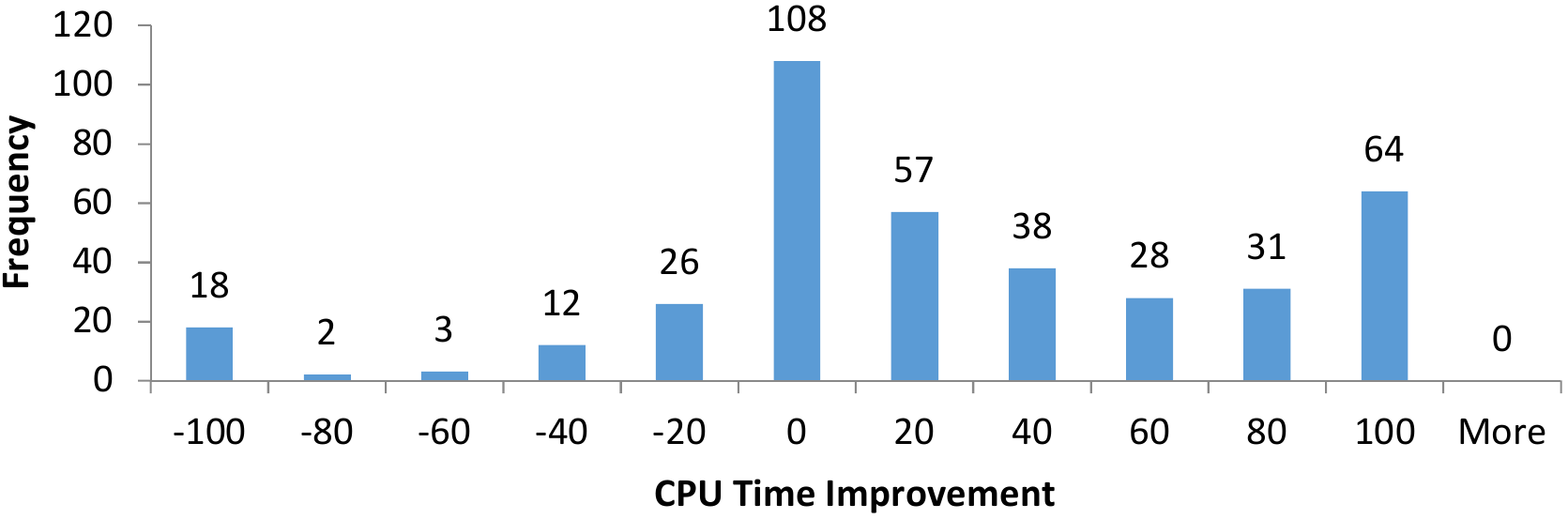}}
\subfigure[Index tuning with feedback.]{ \label{fig:dist:improvement-feedback-20}
\includegraphics[width=0.9\columnwidth]{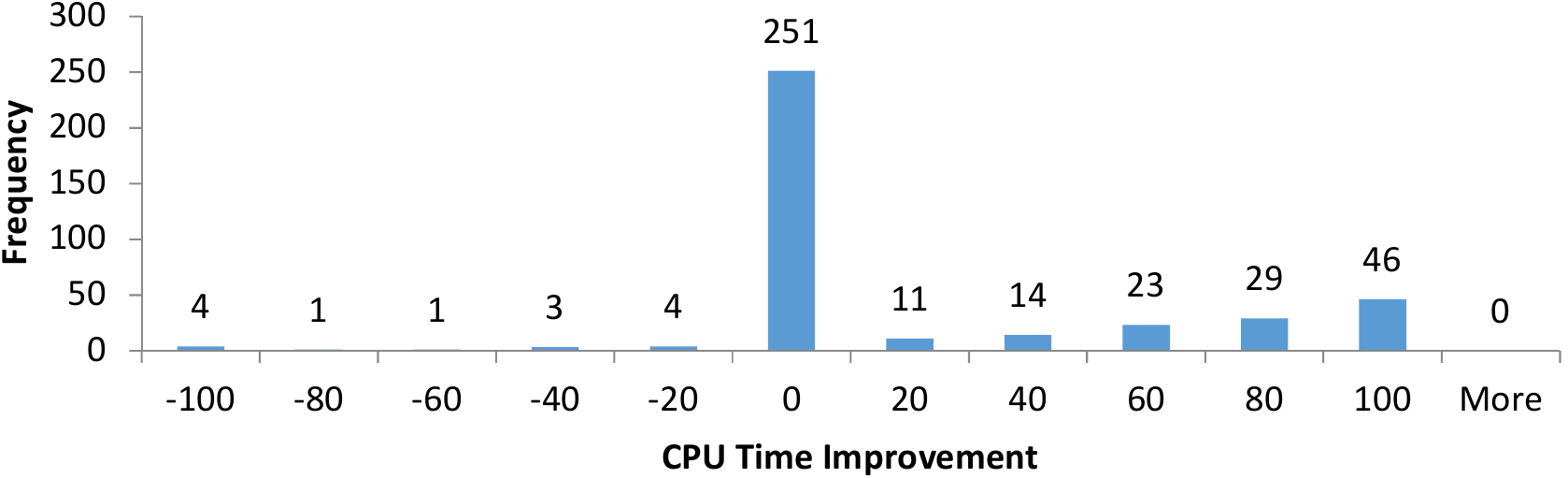}}
\vskip -2ex
\caption{The distributions of CPU time improvement over TPC-DS queries (estimated improvement threshold $\tau = 0.2$).}
\label{fig:dist:cpu-time-TPC-DS-20}
\vskip -3ex
\end{figure}

In Figures~\ref{fig:dist:cpu-time-TPC-DS-0},~\ref{fig:dist:cpu-time-TPC-DS-10}, and~\ref{fig:dist:cpu-time-TPC-DS-20}, we present the distributions of the actual improvement (averaged over all tested configurations) for TPC-DS queries by using execution feedback compared with index tuning without feedback.
In index tuning, usually there is a threshold $\tau$ for \emph{estimated improvement} and a configuration is recommended only if its estimated improvement is above the threshold.
In our experiments, we therefore varied this threshold from 0 to 0.2 (i.e., 20\% estimated improvement).
We have the following observations.
First, using execution feedback in index tuning significantly reduces the chance of query performance regression.
The number of cases with 20\% performance regression (i.e., -20\% improvement) is reduced from 64 to 33 (48.4\% reduction) when $\tau=0$, is reduced from 63 to 21 (66.7\% reduction) when $\tau=0.1$, and is reduced from 61 to 13 (78.7\% reduction) when $\tau=0.2$.

Second, when increasing the estimated improvement threshold $\tau$, the chance of performance regression decreases for both index tuning with and without execution feedback.
However, the chance reduces much faster for index tuning using execution feedback.
This implies that, while index tuning with execution feedback can still misestimate performance improvement, the misestimation is marginal compared to query optimizer's cost estimates.
Such marginal misestimation is more likely to be overcome by slightly increasing the threshold for estimated improvement (20\% is a rule of thumb in practice).

Third, by comparing the figures for index tuning with and without execution feedback, we also observe that estimated improvement is diminished in more cases when execution feedback is utilized.
(Notice that there are more cases in the bin with less than 20\% estimated improvement.)
However, cases with more significant improvement ($\geq$40\%) are less impacted.
In other words, cases falling into the bins with 0\% to 40\% improvement are tend to be moved into the bins with 0\% to 20\% improvement.
Therefore, if performance improvement is indeed significant, index tuning with execution feedback is unlikely to dismiss it.
%Nonetheless, it remains interesting to investigate this shifting pattern and we leave it for future work.

We have observed similar results on the other workloads we tested, though query performance regression is not as significant as we see on the TPC-DS workload.

\vspace{-0.5em}
\paragraph*{Remarks}
Let $X$ be the relative error between estimated cost and actual cost of individual queries.
That is,
$$X=\frac{c(q, I)-a(q,I)}{a(q,I)}=\frac{c(q, I)}{a(q, I)} - 1.$$
We have $c(q,I)=(1+X)\cdot a(q,I)$.
Taking expectation we have
$$E[c(q,I)]=(1+E[X])\cdot E[a(q,I)]=(1+\epsilon)\cdot E[a(q,I)].$$
Since $c(W, I)=\sum_{k=1}^n c(q_k,I)\cdot w_k$, it follows that
$$c(W,I)=\sum_{k=1}^n (1+X_k)\cdot a(q_k, I)\cdot w_k.$$
Taking expectation we have
$$E[c(W,I)]=\sum_{k=1}^n (1+E[X_k])\cdot E[a(q_k, I)]\cdot w_k.$$
Since $E[X_k]=\epsilon$ regardless of $k$, it follows that
$$E[c(W,I)]=(1+\epsilon)\cdot\sum_{k=1}^n E[a(q_k, I)]\cdot w_k=(1+\epsilon)\cdot E[a(W,I)].$$

We therefore have proved the following lemma:
\begin{lemma}
If the expected relative error at query level is $\epsilon$, then the expected relative error at workload level is also $\epsilon$.
\end{lemma}

Now consider the estimated and actual improvements of $I_2$ over $I_1$, in terms of expectation:
$$c(I_1, I_2)=\frac{E[c(q, I_1)]-E[c(q,I_2)]}{E[c(q, I_1)]}=1-\frac{E[c(q,I_2)]}{E[c(q,I_1)]},$$
and
$$a(I_1, I_2)=\frac{E[a(q, I_1)]-E[a(q, I_2)]}{E[a(q, I_1)]}=1-\frac{E[a(q,I_2)]}{E[a(q,I_1)]}.$$
Since $E[c(q, I_1)]=(1+\epsilon)\cdot E[a(q,I_1)]$ and $E[c(q, I_2)]=(1+\epsilon)\cdot E[a(q, I_2)]$, we have
$$c(I_1, I_2)=1-\frac{(1+\epsilon)\cdot E[a(q, I_2)]}{(1+\epsilon)\cdot E[a(q, I_1)]}=1-\frac{E[a(q, I_2)]}{E[a(q, I_1)]}=a(I_1, I_2).$$

Therefore, we have the following observation.
\begin{lemma}
In expectation, the estimated improvement observed at query level carries over to the actual improvement.
\end{lemma}

%\vspace{-1em}
\section{Related Work}\label{sec:relatedwork}

Recently, the problem of estimating query execution time has attracted lots of research attention.
Existing work more or less uses statistical learning techniques~\cite{AhmadDAB11-edbt,AkdereCRUZ12-brown-icde,DugganCPU11,Ganapathi-berkeley09,Li12Robust,WuCHN13,WuCZTHN13}.
The improved cost estimates are useful in a variety of applications, such as admission control~\cite{Tozer-QCop} and query scheduling~\cite{Ahmad-QShufflerVLDBJ11}.

%and in a limited fashion, query optimization~\cite{WuNS16}.

The idea of using execution feedback in query optimization has also been explored in the literature.
Existing work focuses on improving cardinality estimates, using exactly the same cardinality observed in execution (e.g.,~\cite{KabraD98,MarklRSLP04}), statistics built on top of observed cardinality (e.g.~\cite{BrunoC02}), or sampling (e.g.,~\cite{LarsonLZZ07,WuNS16}).
Although this line of work also aims for improving cost estimates (by improving cardinality estimates), its ultimate goal is to impact the decision made by the optimizer so that it may return a different, perhaps better execution plan.
This is different from the goal of the previous work on external cost modeling we have discussed, which does not want to change the query plan.

One noticeable problem when applying feedback, as documented in the literature~\cite{ChaudhuriNR08}, is that \emph{partial} execution feedback may result in \emph{inconsistent} cost estimates that mislead the query optimizer.
That is, if some plans receive improved cost estimates whereas the others do not, then the plan returned by the optimizer might be even worse.
One reason is that, although the optimizer can estimate costs appropriately for query plans with execution feedback, it may underestimate costs for plans without feedback.
%We do not face the same challenge simply because we do not change the plan returned by the optimizer.
%However, we face a similar problem when the index tuning component chooses the best index configuration.
Our analysis in Section~\ref{sec:analysis} indicates that, as long as we have feedback for backbone operators (specifically, when both $\eta$ and $\eta'$ are large), the comparison of plan costs is still reliable.
It is interesting to see the efficacy of our framework when applying it to query optimization.

The problem of index tuning (or, in general, database physical design tuning) has been studied for more than two decades~\cite{FinkelsteinST88}.
Existing index tuning tools use a cost-based architecture that piggybacks on top of query optimizer's cost estimates~\cite{ChaudhuriN97,ValentinZZLS00}.
At the highest level, index tuning tools try to search for an index configuration from a number of candidates that minimizes the total estimated workload cost.
There has been work on different search strategies, such as bottom-up~\cite{ChaudhuriN97}, top-down~\cite{BrunoC05,BrunoC07-merge-reduce}, and constraint-based~\cite{BrunoC08} approaches.
However, to the best of our knowledge, there is no previous work aiming for utilizing improved cost estimates in index tuning.
The framework we studied in this paper can be easily combined with any of the previous search frameworks.
Moreover, there has also been exploratory work towards continuous index tuning~\cite{BrunoC07,SattlerGS03,SchnaitterAMP06}, though such frameworks have not been incorporated into existing index tuning tools as far as we know.

%\vspace{-1em}
\section{Conclusion and Future Work}\label{sec:conclusion}

In this paper, we have studied the problem of operator-level cost modeling with query execution feedback.
We focused ourselves on two major challenges with respect to sparse feedback and mixed cost estimates, and proposed a general framework that deals with both altogether.
We analyzed this framework in detail, and further studied its efficacy in the context of index tuning.

This work opens up a number of interesting directions for future exploration.
First, as we have already briefly discussed in Section~\ref{sec:relatedwork}, it is interesting to further investigate the applicability of our framework in query optimization so that we may be able to obtain better query plan.
Second, although it is straightforward to integrate our framework with existing index tuning architecture, it may have an impact on the search space of index tuning.
In this work, we have only investigated this impact experimentally, and a more formal analysis regarding the change in the search space is desirable if we wish to have a deeper understanding.
Third, although our work is orthogonal to the concrete external cost model being used, in practice different models can make a difference.
There is still lots of work that can be done on model selection and tuning for a given workload.
\iffalse
Fourth, we have shown that different ordering of query execution can have an impact on the short-term effectiveness of our framework, as the predictions made by the models are sensitive to the feedback used in training.
However, it remains an interesting question to understand this impact more deeply.
In particular, how can we know if a model is good enough for use in the framework?
We can further explore the idea of using cross validation in offline training, as we have mentioned in Section~\ref{sec:continuous}.
\fi
All of these are promising directions for future work.

%\newpage
%\clearpage

%\balance
{\small
\bibliographystyle{abbrv}
\bibliography{tuning}
}
%\newpage
%\clearpage

%\appendix

\end{document}